\documentclass{article}
\pdfoutput=1

\usepackage{authblk}

\usepackage{mathtools}
\usepackage{mathpartir}
\usepackage{amsmath}

\usepackage{amsthm}
\usepackage{amssymb}
\usepackage{stmaryrd}
\usepackage{enumitem}

\usepackage[T1]{fontenc}
\usepackage{listings}
\usepackage{xspace}
\usepackage[ruled,algo2e]{algorithm5e}
\renewcommand{\algorithmcfname}{ALGORITHM}
\SetAlFnt{\small}
\SetAlCapFnt{\small}
\SetAlCapNameFnt{\small}
\SetAlCapHSkip{0pt}
\IncMargin{-\parindent}

\newtheorem{definition}{Definition}
\newtheorem{theorem}{Theorem}
\newtheorem{lemma}{Lemma}

\usepackage{graphicx}
\usepackage{wrapfig}
\usepackage{verbatim}
\usepackage{cords}

\newif\ifcomments
\commentsfalse

\ifcomments
\newcommand{\newcomment}[3]{{\color{#1}{[[\textbf{#2:} \emph{#3}]]}}}
\else
\newcommand{\newcomment}[3]{}
\fi

\newcommand{\citeyear}[1]{\cite{#1}}

\title{A Symbolic Logic with Concrete Bounds for Cryptographic Protocols }

\author[A]{Anupam Datta}
\author[B]{Joseph Y. Halpern}
\author[C]{John C. Mitchell}
\author[D]{\\Arnab Roy}
\author[A]{Shayak Sen}

\affil[A]{Computer Science Department, Carnegie Mellon University}
\affil[B]{Computer Science Department, Cornell University}
\affil[C]{Department of Computer Science, Stanford University}
\affil[D]{Fujitsu Laboratories of America}

\date{}

\begin{document}
\maketitle

\begin{abstract}
 We present a formal logic for
quantitative
reasoning about security properties of network
protocols.  The system allows us to derive concrete security bounds that can
be used to choose
key lengths and other security parameters.
We provide axioms
for reasoning about digital signatures and random nonces,
with security properties based on the concrete security of
signature schemes and
pseudorandom number generators (PRG).
The formal logic supports first-order reasoning and reasoning about
protocol invariants,
taking concrete security bounds into account.
Proofs constructed in our logic also provide conventional asymptotic security
guarantees because of the way that concrete bounds accumulate in proofs.
As an illustrative example, we use the formal logic to
prove an authentication property with concrete bounds of a signature-based
challenge-response protocol.
\end{abstract}

\section{Introduction}

Large and complex cryptographic protocols form the backbone of internet
security today.   At the scale of the internet, protocols use keys to
communicate with millions of agents over extended periods of time, with
increasingly powerful adversaries attempting to break them. In this scenario,
evaluating security in a quantitative sense with respect to key lengths and
adversary runtimes is critical.  Concrete security is a practice-oriented
approach to cryptography that requires that proofs of a security property be
accompanied by a concrete probability with which the property is false.

We present Quantitative Protocol Composition Logic (QPCL), a formal system that
supports concrete security proofs of probabilistic safety properties of
cryptographic protocols.  A typical assertion in QPCL is a concrete security
statement of the form $\B^\epsilon(\varphi)$, which we call a \emph{probabilistic belief}. Here $\varphi$ is a temporal
safety property of the execution of a protocol in the presence of probabilistic
polynomial time adversaries, and $\epsilon$ is a probability that is a function
of $\eta$, the security parameter governing the length of keys and nonces, and
$t$, the concrete runtime of the adversary.
This assertion is true if
$\varphi$ holds with probability at least $1 - \epsilon(\eta, t)$. Typically,
$\epsilon$ also depends on functions that represent the probability of an
adversary violating the security of the specific cryptographic primitives used
in a protocol.

QPCL's proof system provides a separation among axioms and rules for
cryptographic primitives, probabilistic beliefs, and programs.  QPCL axioms
about cryptographic primitives such as signature schemes and
pseudorandom-number generators introduce a probabilistic error bound.  These
axioms have formal soundness proofs that show that whenever an axiom is
false on a
trace, an attack on the corresponding primitive can be constructed.  Therefore,
these probabilistic error bounds are closely related to the security of the
underlying primitive.  These soundness proofs are a central contribution of
this paper.  Axioms about cryptographic primitives introduce concrete  bounds
that are combined in QPCL using proof rules for reasoning about
probabilistic beliefs, adapting an approach introduced by Halpern
\cite{Hal37}, which in turn
is based on the \emph{$\epsilon$-semantics} of Goldszmidt, Morris, and Pearl
\citeyear{GMPfull}. Traces generated by protocol programs serve as an instance
of the more general semantic models of Halpern, allowing us to adapt the sound
and complete proof system from~\cite{Hal37} into our logic. Finally, in the
style of Protocol Composition Logic, QPCL's predecessor for symbolic
\cite{DDMST05} and asymptotic reasoning \cite{DDMW06}, QPCL includes a
Hoare-style program modality, and an operator for
temporal ordering of actions that allows reasoning about non-monotonic safety
properties of program traces.

Reasoning about concrete bounds in QPCL necessitates some important differences
from its symbolic and asymptotic predecessors.  Firstly, to achieve precise
concrete bounds, the placement of quantifiers must be carefully considered. For
example, $\forall x.\B^\epsilon(\varphi)$ expresses a probabilistic bound on
$\varphi$ that applies to all choices of $x$, whereas
$\B^\epsilon(\forall x.\varphi)$ bounds the probability of $\forall
x.\varphi$; when $x$ ranges over a large domain, a more precise bound can be derived for the first formula.
Secondly, we restrict the proof system so that only safety properties
are provable in QPCL.
This restriction allows us to follow a useful proof strategy, where we
first reason about the probability of a post-condition of a program
being false, 
and then transfer a bound on that probability to the end of the execution of the protocol. For safety properties the
probability of a formula being false at some point in an execution is bounded by the
probability of a formula being false at the end of that execution.
Section \ref{sec:soundness}
 discusses the importance of the restriction to safety properties.

QPCL assertions are interpreted over traces generated by a simple
probabilistic, concurrent programming language executing in the presence of an
adversary. The semantics of the programming language is independent of any
particular cryptographic primitive, due to a separation between
actions and their implementations. Our formal treatment of this
separation is simpler than
a full module system such as the one used in ~\cite{FournetF711}, and
yet exposes the features (runtimes and trace 
probabilities) needed to define the semantics of probabilistic
trace properties. This separation also allows
the adversary to be represented by any program that runs in polynomial
time, therefore allowing the soundness proofs of our logic to be valid in a
computational model.

We illustrate the use of the QPCL reasoning principles by proving a
``matching-conversations'' style authentication property \cite{DVW1992} for a
simple initiator-responder protocol, which states that after the protocol is
completed, it is known that the precise order in which messages are sent and
received by the participants of the protocol was not altered by an adversary.
Section~\ref{sec:example} contains an outline of the QPCL proof.

Prior work on formal proof techniques for concrete cryptography
\cite{Blanchet06, Barthe09certicrypt,EasyCrypt09} has largely focused on
relational techniques where security is proved via game-based reductions,
similar to traditional cryptographic proofs. Proofs of a number of
cryptographic schemes have been mechanized using these techniques.  However, we
conjecture that reasoning about the security of cryptographic protocols
involves relatively simpler reasoning about trace properties, and that
relational reasoning, while important, does not need to be exposed to the user
of a proof system. In prior work~\cite{RDM07}, indistinguishability
properties such as
secrecy have been proven via stronger trace
properties such as key secrecy. Therefore, relational reasoning in QPCL is pushed down to
the level of the soundness proofs of axioms, which involve a
one-time cost to the designers of the proof system; the users of QPCL
need to reason only about trace properties of protocols.

The rest of this paper is organized as follows. Section~\ref{sec:primer} is a
brief introduction to concrete security. Section~\ref{sec:features} is an
overview of the programming language and the proof system. The programming
language for protocols and its operational semantics is described in
Section~\ref{sec:lang}.  The syntax, semantics, and proof system for QPCL are
described in Section~\ref{sec:qpcl}. Section~\ref{sec:example} has an outline
of the QPCL proof of the initiator-responder protocol considered in the paper.
Section~\ref{sec:related} contains a comparison of QPCL with closely related
work.

\newcommand{\ufcma}{{\mathsf{uf\mbox{-}cma}}}
\newcommand{\prg}{{\mathsf{prg}}}

\section{Concrete Security Primer}
\label{sec:primer}

Concrete security proofs provide bounds for the probability with which an
attacker can subvert some security notion for a primitive. These bounds usually
are a function of
\begin{enumerate}
\item the security parameter $\eta$, which determines the length of keys and
random numbers;\footnote{For ease of presentation, in this paper, we assume
that every cryptographic primitive is governed by the same security parameter $\eta$.}
\item the runtime bound of the adversary, which is often expressed as a
polynomial $tb(\eta)$ of the security parameter; and
\item the number of attempts available to the adversary to break the primitive.
As we describe below, the exact definition of the third parameter is dependent
on the definition of security for that primitive.
\end{enumerate}

In this paper, we study the concrete security of protocols that involve digital
signatures and pseudorandom-number generators; we briefly describe them
here.

A \emph{digital signature scheme} is designed to ensure message unforgeability in a
public key setting. An agent is assigned a private signing key which has an
associated public verification key. A signature scheme $\Pi = (\mathit{Gen}, \mathit{Sign}, \mathit{Verify})$ is specified by a key generation algorithm, $\mathit{Gen}$, a signing
algorithm, $\mathit{Sign}$, and a verification algorithm $\mathit{Verify}$.
$\mathit{Gen}$ generates pairs of public and private keys, $\mathit{Sign}$ produces signatures over messages, and $\mathit{Verify}$ checks that a signature was generated using the signing key. A
characterization of security for a signature scheme $\Pi$ is unforgeability under
chosen message attacks~\cite{GMR88}, shortened to
$\mathsf{uf\mbox{-}cma}$,
which requires that the probability of the adversary creating a new
message/signature pair, given that he can see $q$ signatures produced by $\Pi$
over messages of his choice, is bound by some negligible function
$\epsilon^\ufcma_{\Pi}(\eta, tb, q)$.

A \emph{pseudorandom-number generator} $P$ is an algorithm that produces numbers
that are distinguishable from numbers drawn from a uniformly random
distribution by a computationally bounded adversary with only negligible
probability. This probability is called the advantage of the adversary against
a scheme $P$, and we denote it by $\epsilon^{\prg}_{P}(\eta, tb, q)$, where
$q$ is the number of samples that the adversary sees before deciding
which distribution the samples are drawn from. Psuedorandom number generators
are often used in protocols to generate secrets that are hard to guess by the
adversary.

\section{Overview of the Formal System}
\label{sec:features}

QPCL is a probabilistic program logic with concrete security bounds for
temporal trace properties of cryptographic protocols. A concrete security bound
on an assertion in our logic expresses the probability with which a
computationally constrained adversary can violate the assertion in terms of the
length of keys, the runtime of the adversary and the number of concurrent
sessions of the protocol.

\subsection{Programming Model}

The models on which the logic is interpreted are traces generated by a
simple protocol programming language executing in the presence of an adversary.
The requirements of the logic neccessitate some features of the operational
semantics of the protocol programming language, as the operational semantics
defines how traces are generated from protocol programs interacting in the
presence of an adversary. We now motivate some of these features.

\begin{itemize}

\item \emph{Concrete Runtime.} Reasoning about concrete runtime requires
accounting for the runtime of every computation in an execution. This runtime
includes that of the adversary as well as protocol programs. The operational
semantics is defined by a transition system where transitions are labeled by
these runtime costs.

\item \emph{Concrete Probabilities.} Computing probabilistic bounds requires considering
  probabilities of sets of traces. The transition system is probabilistic and
  the probabilistic branches induce an execution tree, where each path in the
  tree is an execution trace. We use the tree to compute the relevant probabilities.

\item \emph{Computational completeness.} For proofs to be valid in a
computational model, we cannot restrict the adversary to a symbolic model.
However, modeling the adversary explicitly in a Turing-complete language makes
the language very complicated. For proofs of security, which are generally
proofs via reduction to known hard problems, it is possible to ignore the
representation of the adversary. Therefore, we do not explicitly restrict the
adversary's program to a particular language, but reason only about its
input-output behavior.

\item \emph{Concurrency and Adversarial Scheduling.} Programs execute
  concurrently; where the scheduler that switches between programs is controlled
  by the adversary.  The cost
  of the computation required by the adversary to make scheduling decision is
  accounted for in the runtime of a trace.

\end{itemize}

\subsection{Logic and Proof System}
QPCL allows declarative specifications of temporal trace properties, and the
proof system supports reasoning in a manner very similar to first-order
reasoning.  Essentially, axioms related to cryptography have error bounds
associated with them, and these error bounds are propagated throughout the
proof to derive a bound on the final assertion.
To aid reasoning about
cryptographic protocols, QPCL supports the following features:

\begin{itemize}
\item \emph{Temporal Trace Properties.}
Protocols, such as those involving authentication, often require properties
that specify the precise order in which certain events occur. To allow
reasoning about properties such as message ordering, QPCL supports temporal
operators.  Traditionally, such
properties are specified using games~\cite{BR93}, which imply that the security
proofs typically involve the harder task of proving equivalence of games.

\item \emph{Concrete Security Bounds.} Concrete security bounds on assertions
  in QPCL depend on the protocol being reasoned about in a number of ways.
  First, protocols specify the implementations of cryptographic primitives
  to be used; these define the concrete bounds on the axioms related to
  cryptography. Second, concrete bounds on the cryptographic primitives
  usually also depend on the number of queries made to the primitive, and
  to bound this number on any particular trace, we use the number of times
  the cryptographic primitive is called in a protocol.

\item \emph{Safety Properties.} Finally, since concrete security involves
  computing the probability of a property being violated as a function of the
  time available to the adversary, the total runtime of a trace is bounded;
traces end when the total time elapsed is greater than the specified time bound.
Thus, we cannot talk about liveness properties (i.e., things that will
eventually happen) in QPCL.
\end{itemize}

\section{Protocol Programming Language}
\label{sec:lang}

We now describe the syntax and semantics of the protocol programming language
QPCL.

\subsection{Protocol Syntax} \label{section:syntax}

\newcommand{\varbind}{\leftarrow}

\begin{table*}
\[
\renewcommand{\arraystretch}{1.2}
\begin{array}{lrcll}
\mbox{(bitstrings)}& b & ::=   & \agent{A} & \mbox{principal name} \\
                    &   &     & n & \mbox{nonce} \\
                    &	&     & c & \mbox{string constant} \\
                    &   &     & s& \mbox{signature} \\
                    &   &     & \iota & \mbox{session identifier} \\
\mbox{(thread identifier)} & I & ::= & \langle \agent{A},\iota\rangle & \mbox{thread $\iota$ of agent $\agent{A}$} \\
\mbox{(terms)}      & t & ::= & x & \mbox{variable} \\
                    &   &     & b & \mbox{bitstring} \\
                    &	&     & vk(\hat{A}) & \mbox{verification key of agent $\hat{A}$} \\
                    &	&     & I & \mbox{thread-identifier} \\
                    &   &     & \langle t, \ldots, t\rangle & \mbox{tuple} \\
                    &   &     &   \pi_i(t) &\mbox{$i^{\text{th}}$ component of tuple $t$}\\
                    \\
\mbox{(actions)}    & \alpha &::= & \action{send}\;\langle \agent{A}, t\rangle & \mbox{send a term $t$ to \agent{A} } \\
                    &   &    & \action{receive}   & \mbox{receive a term } \\
                    &   &    & \action{new}       & \mbox{generate nonce $n$} \\
                    &   &    & \action{sign}\; t  & \mbox{sign term} \\
                    &   &    & \action{verify}\; \langle t, t, vk(\agent{A}) \rangle & \mbox{verify a signature} \\
	\mbox{(program)}          & P &::= & \irl{stop} &
  												  \mbox{terminated program}\\
&   &     & l\mbox{:}\bind{x}{\alpha}{P} & \mbox{sequence of actions}\\
	\mbox{(role definition)}  & R &::= & r (y) \defeq [P]_I & \mbox{program $P$ with input parameter $y$}\\

\end{array}
\]
\caption{Syntax of the Protocol Programming Language}
\label{tab:languagesyntax}
\end{table*}

\subsubsection{Syntax: Overview}
The protocol programming language is used to represent cryptographic
protocols between principals.  The language is used to describe
the set of roles that comprises a cryptographic protocol. Roles are programs
that honest participants of the protocol run instances of.
For example, the widely used SSL/TLS protocol has two roles for servers and clients.
Each program is a
sequence of actions that represent stateful probabilistic computations on
bitstrings. Examples of actions are network communication and cryptographic
computations such as message signing and verification.
A bitstring may denote a principal name, nonce, string constant, signature,
or a session identifier. As a matter of convention, the metavariable
($\hat{A}, n,$ etc.) that we use to
refer to a particular bitstring denotes its intended usage
(resp. principal name, nonce, etc.).
The syntax of the programming language is summarized in
Table~\ref{tab:languagesyntax} and explained below.

\medskip \noindent \paragraph*{Names, sessions and thread identifiers}
Honest principals can execute multiple instances of different protocol roles.
A thread, which is an instance of a protocol role executed by an honest principal,
is associated with a \emph{thread identifier} $\langle \hat{A}, \iota\rangle$, a pair containing the principal name $\agent{A}$ and a session identifier
$\iota$.  If $I$ is a thread identifier, $\pi_1(I)$ is the principal executing the
thread.\footnote{$\pi_i$ represents the $i^\text{th}$ projection.}
Each thread is assumed to have been assigned an asymmetric signing-verification
key pair.\footnote{We do not reason about key exchange in the logic, and assume keys are pre-assigned before the protocol executes.
}
All threads of a principal $\agent{A}$ share the same key pair,
denoted $(sk(\agent{A}), vk(\agent{A}))$.
A signing key $sk(\agent{A})$ is not
a term in the programming language, since we do not allow signing keys to appear
in messages of a protocol.

\medskip \noindent \paragraph*{Terms, Actions, and Programs}
The \emph{term language} is composed of variables, bitstrings,
verification keys, thread identifiers,
tuples and projections.

\emph{Actions} represent stateful computations on bitstrings that return bitstrings. Examples
of actions are nonce
generation, signing, verification, and communication  (sending and receiving).
An \emph{action statement} $x \leftarrow \alpha$ denotes the binding of
variable $x$ to the result of the action $\alpha$.  We assume that every
bound variable is unique; uniqueness can be ensured by systematically renaming
variables.  Each action statement has a unique label $l$. Actions have the form of
$a(t)$, where $a$ is a name for the action and $t$ is the input.
As a matter of convenience, action labels are not shown explicitly in our examples.

\emph{Programs} are sequences of action statements. We use the notation $[t/x]P$
to denote substituting term $t$ for variable $x$ in program $P$ (renaming
bound variables if necessary to avoid capture).

\medskip \noindent \paragraph*{Role Definitions and Protocols}
A \emph{role} is a program, instantiations of which are run by honest
participants of a protocol. It is defined using a \emph{role definition} of the
form of $r(y) \defeq [P]_I$, which identifies a program $P$ with a name $r$ and
the input parameters $y$ for the role.
Every role has an additional parameter, denoted by the subscript $I$, for
the thread identifier associated with a particular instance of this role.
Every variable in a role is bound, and is either a parameter or assigned
to in an action statement.
A \emph{protocol} specifies a finite list of roles.

\paragraph*{Example Protocol}

\begin{figure}
\[
\begin{array}{lcl}
    A \rightarrow B & : & m \\
    B \rightarrow A & : & n, \sig{B}{\resp, n, m, A} \\
    A \rightarrow B & : & \sig{A}{\init, n, m, B}
\end{array}
\]
\caption{Challenge-response protocol as arrows-and-messages}
\label{figure:crarrows}
\end{figure}

The two roles of the challenge-response protocol depicted informally in
Figure~\ref{figure:crarrows} are written formally in
Figure~\ref{figure:crcords}, with \cordname{Init_{CR}} the initiator role and
\cordname{Resp_{CR}} the responder role.  The initiator role
\cordname{Init_{CR}} begins with a \emph{static parameter} \agent{Y}, which is
the name of the agent that the initiator intends to communicate with.  If an
agent Alice, for example, uses \cordname{Init_{CR}} to authenticate an agent
Bob using the CR protocol, then Alice executes the instance
\cordname{Init_{CR}(Bob)} of the initiator role obtained by substituting Bob
for the static parameter \agent{Y}.
The value of the static parameters (the choice of the intended responder in this case) is chosen by the
adversary before a role is executed.

In words, the actions of a principal
executing the role \cordname{Init_{CR}} are: generate a fresh random number;
send a message with the random number to the peer \agent{Y}; receive a message
with source address \agent{Y}; verify that the message contains \agent{Y}'s
signature over the data in the expected format; and finally, send another
message to \agent{Y} with the initiator's signature over the nonce sent in the
first message, the nonce received from \agent{Y}, and \agent{Y}'s identity.
The subscript $X$ on the square brackets $[ \ \ldots \ ]_X$ enclosing the
actions of \cordname{Init_{CR}} indicate that $X$ is the
thread-identifier of the thread
executing these actions. $X$ is a pair $\langle\hat{X}, \iota\rangle$ for some
principal $\agent{X}$, and thread name $\iota$.  A role $[ \ \ldots \ ]_X$ may
include an action that signs a message with the private key of $\agent{X}$, but
not with the private key of $\agent{Y}$ unless $\agent{Y}=\agent{X}$.

To simplify the presentation, we use pattern-matching notation such $(\langle x_1, x_2
\rangle \varbind \alpha; P)$ for binding more than one variable in an action
statement involving tuples. This should be read as the program with the
variables in the pattern replaced with the correct component of the tuple in
the rest of the program. For example, the above program should be read as
shorthand for $(x \varbind \alpha;[ \pi_1(x)/x_1,\pi_2(x)/x_2]P)$. We also use
$\_$ in place of a variable when
the return value from an action is never used.

\begin{figure}[th]
\begin{minipage}{200pt}
\begin{eqnarray*}
  \cordname{Init_{CR}}(\agent{Y}) & \defeq & [\\
    & & m \varbind \new;\\
    & & \_ \varbind \send   \langle\agent{Y}, m \rangle;\\
    & & \langle \langle y, s\rangle, N' \rangle \varbind \irl{receive};\\
    & & \_ \varbind \action{verify} \ \langle s, \langle``{\tt Resp}", y, m, \agent{X}\rangle, vk(\agent{Y})\rangle ;\\
    & & r \varbind \action{sign} \langle``{\tt Init}", y, m,  \agent{Y}\rangle
    ;\\
    & & \_ \varbind \send \langle \agent{Y}, r \rangle; \\
    & & \stp \\
  &&  ]_{{X}}
\end{eqnarray*}
\end{minipage}

\begin{minipage}{200pt}
\begin{eqnarray*}
\cordname{Resp_{CR}}() & \defeq & [\\
    & & \langle x, \agent{X}\rangle \varbind \irl{receive}; \\
    & & n \varbind \new;\\
    & & \ r \varbind \action{sign} \langle``{\tt Resp}", n, x, \agent{X}\rangle;\\
    & & \_ \varbind \send \langle \agent{X}, \langle n, r\rangle\rangle; \\
    & & \langle t,N'' \rangle \varbind \irl{receive}; \\
    & & \_ \varbind \action{verify}  \langle t, \langle``{\tt Init}", n, x, \agent{Y}\rangle,
    vk(\agent{X})\rangle ;\\
    & & \stp \\
  &&  ]_{{Y}}
\end{eqnarray*}
\end{minipage}
\caption{Roles of the Challenge-response protocol}
\label{figure:crcords}
\end{figure}

\subsection{Semantics}

The semantics of the protocol programming language describe how protocol
programs interact in the presence of an adversary who is allowed to execute
arbitrary polynomial-time computations.

\subsubsection{Semantics: Overview}

The semantics describe the interleaved execution of a collection of protocol
programs to produce traces.  At a given point in time, at most one protocol
program in a collection is active, while other programs wait to be executed.
Recall that a program is a list of actions. We describe how each action is
executed and how actions in different programs are interleaved.

Individual actions are executed by invoking stateful probabilistic computations
on bitstrings called \emph{action implementations}. For example, informally, to
execute a signature-signing action, the input to the $\irl{sign}$ action is
passed to an implementation of digital signature schemes that maintains its own
state. Similarly, to execute the sending of a message, the input to the
$\irl{send}$ action is sent to the adversary, who is responsible for relaying
messages. After the computation implementing the action has finished executing
and returns a value, the remainder of the active program continues execution.
Certain actions are allowed to make scheduling decisions instead of simply
returning values.  For example, the receive action is allowed to signal the
active program to wait, simulating the adversary's power to make scheduling
choices when a program is waiting to receive a message. A special computation,
the \emph{scheduler}, is invoked when all threads are waiting.

The adversary controls the scheduler and the send and receive actions, and is
allowed to run arbitrary computations when these actions are invoked. In
effect, the adversary controls network communication and the scheduling of protocol
threads.\footnote{The adversary is also allowed to return garbage values after
executing actions that are not parsable as terms, but such values are ignored.
Therefore, we make the assumption that all return values of actions are
parsable as terms.} However, the time for executing every action in a trace is
recorded and the total runtime of every trace is bounded by a polynomial of the
security parameter, which corresponds to a polynomially bounded adversary.

\emph{Action transitions} represent the behavior of action implementations, the
stateful probabilistic computations that implement actions.  Action
transitions take as input a state for the action implementation and an
input bitstring and return a distribution over updated states and output
bitstrings.

Terms such as projections of pairs are resolved down to values representing
bitstrings before an action is invoked on such a term. Such stateless
computations are represented by the \emph{term-evaluation}
relation.
Finally, action transitions along with the term-evaluation relation
are used in defining \emph{configuration transitions} that
describe how individual actions are executed and different protocol
programs are interleaved.

\subsubsection{Semantics: Details and Examples}

We now describe term-evaluation and action transitions,
and then define configuration transitions formally.

\newcommand{\impl}[1]{\langle\langle #1 \rangle\rangle}

\paragraph*{Term Evaluation and Values} Terms represent
deterministic, effect-free computations on bitstrings.  \emph{Values} are closed
terms that cannot be simplified further, namely, bitstrings and tuples
of bitstrings. 
The evaluation of term $t$ to value $v$ with cost $c$ is denoted by
the term-evaluation relation $\eval{t}{c}{v}$. For example, we have that
$\eval{\pi_1(\langle A, B\rangle)}{c}{A}$, for some cost $c$.

\paragraph*{Implementations and Implementation Transitions} \emph{
Implementations} are stateful probabilistic computations on bitstrings
corresponding to each action in the programming language. A special implementation, the scheduler $Sr$,
does not correspond to any action, but is responsible for scheduling when
no programs are active.

\emph{Implementation transitions} are
probabilistic transition relations that describe the behavior of action
implementations.  The observable behavior of implementations includes the cost
of computing a particular result and the probability with which the particular
result is computed. The following transition relation represents one outcome of
executing implementation $A$ on input $t$:

\[\sigma,A(t) \xmapsto{c,p} \sigma',ra.\]

Here, $\sigma$ denotes the state maintained by $A$, $c$ denotes the cost of the
outcome, $p$ denotes its probability, $\sigma'$ denotes the output state and
$ra$ denotes the result of the computation (called the \emph{reaction}), which
has the following form: \[ra::= \ret{t} \mid \abort \mid \wait \mid \switch{I}
. \] The reaction can either return a value of the expected type
($\irl{ret}(t)$), signal the active program to abort ($\abort$) or wait
($\wait$), or force a context switch to the thread with identifier $I$ ($\irl{switch}[I]$). For
example, an implementation of the signature verification action may return if
the verification succeeds; otherwise, it can signal the thread to stop.
The state maintained by an implementation is treated abstractly by the
operational semantics.
We have deliberately left the low-level details of states
unspecified.  The soundness proofs do not depend on the specifics of
the implementations.
Moreover, adversaries modeled using such an abstract representation are allowed
to implement arbitrary data structures to mount efficient attacks on the
protocol.
To build intuition, we now provide examples of simple instances of
implementation transitions.

Consider a deterministic implementation $Sign$ of the
signature action that keeps a count only of the number of messages it has
signed as its state. In this case, the action transition has the
following form:
\[n, Sign(m)
  \xmapsto{c,1} n+1,\ret{s}.\]

Next, consider an example where an implementation $\mathit{Verify}$ of the
signature verification action $\irl{verify}$ fails to verify a
signature:
\[\triv,\mathit{Verify}(\langle m, m', vk(N)\rangle) \xmapsto{c,1} \triv,\irl{abort}.\]
The implementation $\mathit{Verify}$ does not
maintain any state; its state is thus denoted $\triv$.

Now consider an implementation $Send$ of the $\irl{send}$ action whose state is
a set $s$ of sent messages that corresponds to an adversary who randomly drops
half the messages:
\[\begin{array}{lcl}
  s, Send(\langle \hat{A}, m\rangle)
&\xmapsto{c,0.5}& s,\ret{\triv}\\
  s, Send(\langle \hat{A}, m\rangle)
&\xmapsto{c,0.5}& s \cup \{\langle \hat{A}, m\rangle\},\ret{\triv}.
\end{array}\]
In the first outcome above, the implementation $Send$ keeps the set
of sent messages unchanged, effectively dropping the message $m$.
In the second outcome, $Send$
updates the state $s$ with the new message. In both outcomes, the
null tuple $\triv$ is returned.

Consider an implementation $Receive$ of the $\irl{receive}$
action that signals the active program to wait when it has no messages to relay:
\[\{\}, Receive(\triv) \xmapsto{c,1} \{\}, \wait.\]

Finally, $\switch{I}$ is used by the scheduler to indicate a switch to the thread with identifier $I$. For example,
a scheduler $Sr$ that manages a list of $n$ threads and picks one at
random has the
following transition for each $k$:
\[\{I_1 \cdots I_n\},Sr \xmapsto{c,\frac{1}{n}} \{I_1 \cdots I_n\}, \switch{I_k}.\]

We require that the number of possible transitions from a given state be
finite and that the transition probabilities sum to $1$. For example, if for
the state $\sigma$, implementation $A$, and input $t$, the possible
transitions are
\[\sigma,A(t) \xmapsto{c_1,p_1} \sigma_1',ra_1 \quad \ldots \quad
  \sigma,A(t) \xmapsto{c_n,p_n} \sigma_n',ra_n, \]
then we require that  $\sum_{i=1}^{n}  p_i = 1$ and that for all $i$, $p_i > 0$.

\paragraph*{Setup} A setup $\mathcal{S}$ is a mapping $\impl{.}_I$ that associates an action $a$ with an action implementation $\impl{a}_I$ for the thread with identifier $I$, along with an implementation
$Sr$ for the scheduler. We will see below that only the part of the setup that
pertains to cryptography is specified by a protocol.

\paragraph*{Restrictions on Actions}
 \begin{table}
\[\renewcommand{\arraystretch}{1.2}
  \begin{array}{ l  l  l }
  \hline
  \textsf{\textbf{Action}} & \textsf{\textbf{Type}} & \textsf{\textbf{Reactions}} \\\hline
  \action{new} & \ttype{unit} \rightarrow \ttype{nonce} &  \ret{t} \\
  \action{send} & \ttype{principal\_name} \times \ttype{message} \rightarrow \ttype{unit} &  \ret{t} \\
  \action{receive} & \ttype{unit} \rightarrow \ttype{principal\_name} \times \ttype{message} &  \ret{t}, \wait \\
  \action{sign} & \ttype{message} \times \ttype{sgn\_key} \rightarrow \ttype{message} &  \ret{t} \\
  \action{verify} &  \ttype{message} \times \ttype{message} \times \ttype{sgn\_key} \rightarrow \ttype{unit} &  \ret{t}, \abort \\
  \hline
\end{array}\]
\caption{Types and permitted reactions for actions}
\label{tab:actions}
\end{table}

We make well-typing assumptions about the reactions of action implementations.
These assumptions are easily discharged, since straightforward checks
can identify when responses are not the expected type, and ignore them.
Associated with each action is a type and a set of allowed reactions;
these are summarized in Table \ref{tab:actions} and described here. For instance,
according to the table, the only reaction allowed for action
$\snd{\langle \hat{A}, t\rangle}$ is $\ret{\triv}$;
a context switch is not allowed on a send.
The implementation for $\action{receive}$
can respond with $\ret{t}$ or $\wait$. This means that a context switch
can be enforced during a receive action. The fact that only $\action{receive}$
can respond with a $\irl{switch}$ means that each thread executes its
program uninterrupted until it reaches a $\action{receive}$ action; at this
point,  the adversary has the opportunity to schedule other threads.

The action $\irl{new}$ would
typically be implemented by a pseudorandom-number generator, and can
respond only with $\ret{t}$.

The action $\irl{sign}$ returns $\ret{t}$;
$\irl{verify}$ may respond with a $\ret{\triv}$ on success or $\abort$ on
failure.
These actions would typically be implemented by a digital signature scheme.

An implementation may either be specified by the protocol setup or be provided by the adversary.
For cryptographic protocols, it may be reasonable to expect that
implementations for $\irl{send}$ and $\action{receive}$ are provided by the adversary,
as it controls the communication channel and is responsible for scheduling
processes. The implementations for $\irl{sign}$, $\irl{verify}$, and $\irl{new}$
are specified by the protocol according to its choice of signature algorithm and
pseudorandom-number generator.

\paragraph*{Configuration Transitions}

Finally, configuration transitions describe how protocol programs are executed
and interleaved to produce traces.
A configuration is a snapshot of the
execution of all computation in our system. This snapshot includes each of
the protocol programs and the state of each implementation.
Formally, a \emph{configuration} $\conf$ has the form
$\config{E}{r}{\traces}{T}$ . $\traces$ is a set of waiting threads, where each
thread has the form $I:P$, consisting of $I$, the thread-identifier, and $P$,
the program. $T$, the thread after the wedge, is the active
program being processed. When no program is active, a configuration is denoted
 $\config{E}{r}{\traces}{\cdot}$.  $E$, known as the \emph{global state}, is
the collection of the states for the action implementations.

\paragraph*{Global State} The structure of the global state $E$ captures
which action implementations are allowed to share information.  For example,
the adversary is allowed to combine any information that he infers from
scheduling or that he receives from the network. However, we require that the
pseudorandom-number generators from different threads operate on different states.
The global state is a set that has the form $\{h_1: \sigma_1, \cdots, h_k:
\sigma_k\}$, where each $\sigma_i$ is a local state for some action
implementation. The label $h_i$ determines which implementation state
$\sigma_i$ belongs to. Implementations that share state have the same
label in the
global state. We need to define for each action how the label for that action
in the global state is determined.  The label of an action $a$ in a thread
with identifier $I$ in the global state is denoted $\sttag{a_I}$. If
the implementations of action
$a$ in thread $I$ and action $b$ in thread $I'$ operate on the same state, then
we want that $\sttag{a_I} = \sttag{b_{I'}}$. Otherwise, if the actions operate
on separate states, then we want  $\sttag{a_I} \not= \sttag{b_{I'}}$.  The
$\irl{send}$ and $\irl{receive}$ actions, which are implemented by the
adversary, share state for every thread. So, for all $I$, we choose
$\sttag{\irl{send}_I} = \sttag{\irl{receive}_I} = \irl{a}$, for a constant
symbol $\irl{a}$. On the other hand, the $\irl{sign}$ and $\irl{verify}$
actions have separate state for each principal, so we parameterize
the label of these actions with a principal, denoting them $\sttag{\irl{send}_I} =
\sttag{\irl{verify}_I} = \irl{s}(\pi_1(I))$ for a symbol $\irl{s}$ (recall
that $\pi_1(I)$ is the principal corresponding to $I$). The action $\irl{new}$
has a separate state for each thread, so we denote its label
$\sttag{\irl{new}_I} = \irl{p}(I)$, for a symbol $\irl{p}$.  Finally, the
scheduler $Sr$ is implemented by the adversary, so $Sr$'s label in
the global state is $\irl{a}$, the same as that of $\irl{send}$ and
$\irl{receive}$.

Configuration transitions have the form $\ctrans{\conf}{c,p}{}{\conf'}$, where $c$ and $p$ denote the runtime
cost and probability of the transition. Transitions where an action
executes to return some value carry
a record of the action below the arrow. This record has the form $(I, a(v), v')$, where $I$ is the active
thread, $a(v)$ is the action executed and $v'$ is the value returned. (See rule $\rulen{act}$ below.)

We now present rules for the evolution of configurations. These rules are parametric in the transitions for evaluation of actions and terms.
Table \ref{tab:transitionrules} presents the transition rules for
configurations.\footnote{In the transition rules, the notation $\Psi, I:P$ denotes
$\Psi\cup\{I:P\}$; similarly,
$E, \sttag{a_I}:\sigma$ denotes $E\cup\{\sttag{a_I}:\sigma\}$.}

\begin{table*}
  \begin{mathpar}
  \infern{
    \eval{t}{c'}{v}\\
    \sigma,\impl{a}_I(v) \xmapsto{c,p} \sigma',\ret{v'}
  }{
    \ctrans{\config{E,\sttag{a_I}:\sigma}{r}{\traces}{I:\bind{x}{a(t)}{P}}}
    {c+c',p}
    {(I, a(v), v')}
    {\config{E,\sttag{a_I}:\sigma'}{r'}{\traces}{I:[v'/x]P}}
  }{act}

  \infern{
    \eval{t}{c'}{v}\\
    \sigma,\impl{a}_I(v) \xmapsto{c,p} \sigma',\abort
  }{
    \ctrans{\config{E,\sttag{a_I}:\sigma}{r}{\traces}{I:\bind{x}{a(t)}{P}}}
    {c+c',p}
    {}
    {\config{E,\sttag{a_I}:\sigma'}{r'}{\traces}{\cdot}}
  }{abort}

  \infern{
    \eval{t}{c'}{v}\\
    \sigma,\impl{a}_I(v) \xmapsto{c,p} \sigma',\wait
  }{
    \config{E,\sttag{a_I}:\sigma}{r}{\traces}{I:\bind{x}{a(t)}{P}}
    \xrightarrow{c+c',p}
    \config{E,\sttag{a_I}:\sigma'}{r'}{\traces,I:\bind{x}{a(t)}{P}}{\cdot}
  }{wait}

\infern{
  \sigma,Sr \xmapsto{c,p} \sigma',\switch{I}
  }{
    \ctrans{\config{E,\sttag{Sr}:\sigma}{r}{\traces,I:P}{\cdot}}
    {c,p}
    {}
    {\config{E,\sttag{Sr}:\sigma'}{r'}{\traces}{I:P}}
  }{switch}

\end{mathpar}
\caption{Operational Semantics for the Protocol Language}
\label{tab:transitionrules}
\end{table*}

The first three rules, $\rulen{act}$, $\rulen{abort}$, and $\rulen{wait}$,
describe how a single action in the active program is executed.  In each of
these rules, $\impl{a}_I$ is the implementation of the action $a$ for thread
$I$, as per the setup.  Each rule describes how a statement is executed. First
the input to the action is evaluated to a value $v$; then the implementation
$\impl{a}_I$ is run on $v$ to yield a new state $\sigma'$ and a reaction, and
the change in state is reflected in the global state.  According to rule
$\rulen{act}$, when $\impl{a}_I$ returns a value $v'$, the active program
$\bind{x}{a(t)}{P}$ transitions to $[v'/x]P$, where all occurences of the
variable $x$ are replaced with $v'$. For example, for the signature action
$\irl{sign}$, when the implementation $\impl{\irl{sign}}_I$ is
called on some message $m$, whe implementation returns a
signature for $m$.
On the other hand, when the implementation responds with $\abort$,
the current program is stopped. Similarly, when the implementation
responds with $\wait$, the active program is moved to the set
of waiting programs.

The rule $\rulen{switch}$ describes how
different programs in a configuration are interleaved. When an implementation responds with
$\irl{switch}[I']$, the thread $I'$ is made the active program.

In each rule, the total computation cost and the probability are recorded in
the transition.  The transitions due to $\rulen{act}$ contain a record of the
action executed in the transition and the value returned. Such transitions will
be referred to as \emph{labeled transitions}.

\newcommand{\cinit}{\conf_{\mathsf{init}}}
\paragraph*{Execution Traces} A \emph{trace} $\trace$ is a sequence of
configuration transitions.  A trace records the actions in the execution of the
protocol, the runtime cost incurred, and its probability.  We say that the
action label $(I,a(t),t') \in \trace$ if $\trace$ contains a labeled transition
$\ctrans{\conf}{c,p}{(I, a(t), t')}{\conf'}$. We also define a function
$rt(\trace)$, which returns the total cost of execution of the trace (i.e., the
sum of the runtime costs for each transition), and a function $p(\trace)$,
which returns the probability of the trace (i.e., the product of the
probabilities for each transition). The relation $\trace_1 \preceq \trace_2$
denotes that $\trace_1$ is a prefix of $\trace_2$.

\paragraph*{Execution Trees}
An execution tree $\tree$ rooted at a configuration $\conf$
is either
the configuration $\conf$,
or has the form $(\conf \xrightarrow{p_1,c_1}  \tree_1 ~|~\cdots~|~\conf
  \xrightarrow{p_k,c_k}  \tree_k)$, where
$\conf \xrightarrow{p_1,c_1} \conf_1, \cdots \conf
  \xrightarrow{p_k,c_k} \conf_k$ are valid transitions from $\conf$
  and $\tree_1, \cdots, \tree_k$ are execution trees rooted at
  $\conf_1, \cdots, \conf_k$, respectively.

Informally, a \emph{tree} $\tree$ is a branching structure
of configuration transitions.
   We write $\trace \in \tree$ if $\trace$ is a path in $\tree$.
Every path through a tree is a trace.
We
extend $p$ to finite trees by taking
$\measure{\tree} = \sum_{\trace\in\tree}\measure{\trace}$.

\paragraph*{Initialization and Termination}
For a protocol with roles $\{R_1, R_2, \ldots, R_k\}$,
the \emph{instance vector} $\vec{n} = \langle
n_1, n_2,\ldots, n_k \rangle$ denotes the number of instances of each
of the $k$ protocol roles; the \emph{initialization vector}
$\vec{\mathbf{t}} = \langle \vec{t_1}, \vec{t_2},\ldots, \vec{t_k} \rangle$
is a vector of vectors, where $\vec{t_i}$ is a list of values
corresponding to the input
parameters of all instances of the $i$th role.  Given an initialization vector,
the \emph{initial configuration} consists of all the instances of protocol
roles running on each thread, where values from the initialization
vector are substituted for the corresponding parameters.

\paragraph*{Protocol} A protocol $\Q$, is a pair $(\vec{R},
\mathcal{S}_{\mathcal{Q}})$, where $\vec{R}$ is a list of protocol roles, and
$\mathcal{S}_{\mathcal{Q}}$ is a mapping from actions to implementations for actions
corresponding to cryptography, namely, $\irl{sign}$, $\irl{verify}$, and
$\irl{new}$.

\paragraph*{Adversary} An adversary $\mathcal{A}$ is a pair
$(\vec{\mathbf{t}}, \mathcal{S}_{\mathcal{A}})$, where $\mathcal{S}_{\mathcal{A}}$
is the setup for the $\irl{send}$ and $\irl{receive}$ actions and the
scheduler, and $\vec{\mathbf{t}}$ is the initialization vector.
The remainder of the setup, in particular, the implementations for
the signature and pseudorandom-number generation actions, is specified by
the protocol.

Given a runtime bound $tb(\eta)$ that is a polynomial of the security
parameter $\eta$, a trace $\trace$ is considered \emph{final} if all threads are
terminated or if there exists a transition from the last state such that the
total runtime exceeds $tb(\eta)$. A tree $\tree$ is considered final if all
traces in $\tree$ are final.

\begin{definition}
\label{def:comp-trace}
(Feasible Traces) For a protocol $\mathcal{Q}$, setup $\mathcal{S}$,
instance vector $\vec{n}$,
initialization vector $\vec{\mathbf{t}}$, and polynomial $tb(\eta)$ that
bounds the running time,
$T_{\mathcal{Q}}(\mathcal{S}, \mathbf{t}, \vec{n}, tb, \eta)$ is the final
 tree obtained starting from the initial configuration.
Given an adversary $\mathcal{A}$ that provides implementations
$\mathcal{S}_\mathcal{A}$ and initial values $\vec{\mathbf{t}}$, and a protocol
$\mathcal{Q}$ that provides implementations $\mathcal{S}_\mathcal{Q}$, the
execution tree is written as $\T_{\mathcal{Q},\mathcal{A}}(\eta,tb,\vec{n}) =
T_{\mathcal{Q}}(\mathcal{S}_\mathcal{A}\cup\mathcal{S}_\mathcal{Q},
\vec{\mathbf{t}}, \vec{n}, tb, \eta)$.
\end{definition}

The function $\T_{\mathcal{Q},\mathcal{A}}$ serves as the model on which QPCL formulas
are interpreted.

\section{Quantitative Protocol Composition Logic}
\label{sec:qpcl}

In this section, we present the syntax, semantics, and proof system for
Quantitative Protocol Composition Logic (QPCL).  While the logic is similar to
Computational PCL~\cite{DDMST05}, the key technical difference is that we can
specify and reason about exact security trace properties in QPCL. Specifically,
the axioms and proof rules of the logic are annotated with exact bounds,
thereby enabling an exact bound to be derived from a formal axiomatic proof of
a security property.

QPCL consists of two kinds of formulas: basic formulas that do not involve
probability and are interpreted on execution traces; and conditional
formulas of the form of $\varphi_1 \rightarrow^\epsilon \varphi_2$ that are
annotated with a probability function $\epsilon$ and are interpreted on
execution trees. The probability function $\epsilon$ is parameterized
by the security parameter, adversary runtime, and number of sessions of
the protocol, and specifies, for a particular setting of these parameters,
the probability of the conditional implication not holding.

\subsection{Syntax}
\label{section:qpcl-ss}
\newcommand{\progr}{\mathit{program}}

\begin{table}[t]
\[
\begin{array}{lcl}
\multicolumn{3}{l}{\mbox{\textbf{Action Predicates:}}} \\
\pred{a} & ::= & \pred{Send}(I, t)\,|\, \pred{Receive}(I, t)\,|\,
\pred{Sign}(I, s, t)\,|\, \\
& & \pred{Verify}(I, s, t, \hat{B})\,|\, \pred{New}(I, t)\\ \\
\multicolumn{3}{l}{\mbox{\textbf{Basic Formulas:}}} \\
\varphi & ::= & \pred{a} \,|\, \pred{a} < \pred{a}'\,|\, t=t' \,|\, \pred{Start}(T) \,|\,
\pred{Contains}(t,\, t') \,|\, \\
& & \pred{Honest}(\agent{B}) \,|\,  \varphi\,\wedge\, \varphi'  \,|\, \neg \varphi \,|\,
\varphi\,\left[\progr\right]_{I}\, \varphi' \,|\, \forall x. \phi \\ \\
\multicolumn{3}{l}{\mbox{\textbf{Conditional formulas:}}} \\
\psi & ::= & \phi  \,|\, (\phi\,\pimp{\epsilon}\, \phi')
\,|\, \psi \land \psi' \,|\, \neg\psi \,|\, \forall x. \psi

\end{array}
\]
\caption{Syntax of the logic}
\label{table:logic-syntax}
\end{table}

The syntax of formulas is given in Table~\ref{table:logic-syntax}.   We
summarize the meaning of formulas informally below, with precise semantics
presented subsequently.

Basic formulas do not involve probabilities, and are interpreted over execution
traces.
We use the same syntax for terms and variables in formulas as for
terms and variables in the protocol programming language.
For every protocol action, there is a corresponding action predicate
that asserts that the action has occurred on the trace.  For example,
$\pred{Send}(I, t)$ holds on a trace where the thread with thread-identifier $I$ has sent the term
$t$, while $\pred{New}(I, t)$ holds on a trace where the thread $I$ has created
the term $t$ using the random nonce generation action.  Similarly,
$\pred{Sign}(I, t, t')$ holds when a signature $t$ for message $t'$ has been
generated, and $\pred{Verify}(I, t, t', \agent{B})$ holds when $t$ has been
verified to be $\agent{B}$'s signature over $t'$.  $\pred{Contains}(t,t')$
holds on a trace if $t$ can be derived from $t'$ using symbolic actions, such
as projection and message-recovery from a signature. $\pred{Honest}(\agent{B})$
holds if $\hat{B}$ is a principal associated with a protocol thread.  We use
$\pred{Fresh}(I, t)$ as syntactic sugar for the formula $\pred{New}(I,t) \wedge
\forall t' (\pred{Send}(I,t') \rimp \neg \pred{Contains}(t,t'))$, which means
that the value of $t$ generated by $I$ is ``fresh'' in the sense that $I$
generated $t$ using a $\action{new}$ action and did not send out any message
containing $t$. We use $\pred{FirstSend}(I, n, m)$ as syntactic sugar for the
formula $ \pred{Contains}(n, m) \wedge \forall t. ((\pred{New}(I, n) \wedge
\pred{Send}(I, t) < \pred{Send}(I, m)) \rimp \neg \pred{Contains}(n, t))$.

The logic includes modal formulas of the form $\phi_1[\progr]_I\phi_2$, where
$\progr$ is a fragment of some protocol's program. This formula holds on a
trace $\trace$ if, roughly speaking, whenever $\trace$ can be split into three
parts, $\trace_1$, $\trace_2$, and $\trace_3$ (so that $\trace =
\trace_1;\trace_2;\trace_3$) such that $\trace_2$ ``matches'' $\progr$ and
$\phi_1$ holds on trace $\trace_1$, then $\phi_2$ holds on $\trace_1;\trace_2$.

The logic also includes standard connectives and quantifiers of first-order
logic. In addition, we have conditional formulas of the form
$\varphi\,\pimp{\epsilon}\, \varphi'$. As we said in the introduction,
$\epsilon$ here is a function.  Its arguments are $\eta$ (the security
parameter), $tb$ (the adversary's time bound), and $\vec{n}$ (the number of
instances of each role of the protocol being analyzed).  The formula
$\varphi\,\pimp{\epsilon}\, \varphi'$  holds if the conditional probability of
$\varphi'$ given $\varphi$ is greater than or equal to $1 - \epsilon(\eta, tb,
\vec{n})$, for every setting of $\eta$, $tb$ and $\vec{n}$.  Conditional
formulas are interpreted over execution trees parameterized by $\eta$, $tb$ and
$\vec{n}$. While the syntax of QPCL allows arbitrary functions of $\eta$, $tb$
and $\vec{n}$ to be used in conditional formulas, we will see in
Section~\ref{section:proof-system} that formulas provable using QPCL's proofsystem use a
very restricted set of functions, which is the set of linear combinations of
the bounds on QPCL's axioms.

We abbreviate a special case of the conditional formula of the
form $true \rarrow^\epsilon \varphi$ as $B^{\epsilon} \varphi$.  Such formulas
will prove useful for specifying exact security properties of cryptographic
primitives (such as signature schemes) and protocols (such as matching
conversation-style authentication properties~\cite{BR93}).

\newcommand{\traceset}{\mathcal{T}}
\newcommand{\lamhat}{\hat{\lambda}}
\newcommand{\joecom}[1]{\newcomment{red}{Joe}{#1}}
\subsection{Semantics}
\label{sec:semantics}

Basic formulas are interpreted over traces, whereas
conditional formulas are interpreted over execution trees.

For basic formulas, we define the semantic relation
$\trace, \lambda \models \varphi$ denoting that basic formula $\varphi$ holds on trace
$\trace$, where $\lambda$ is a \emph{\valuation} that maps variables to values. The
notation $\lambda[x\mapsto v]$ denotes the \valuation that is
identical to $\lambda$ except that it maps
variable $x$ to value $v$. Recall that all bound
variables are assumed to be unique, as uniqueness can be ensured via renaming.
The lifting of $\lambda$ from variables to terms is denoted $\lamhat$.
$\lamhat(t)$ yields a value after replacing each variable in $t$ by $\lambda(x)$.
We write $\trace \models \varphi$ if, for all $\lambda$, $\trace, \lambda \models
\varphi$.

For conditional formulas, we define the semantic relation $\T, (\eta,tb,\vec{n}), \lambda \models \psi$. Here, $\T$ is a function that, like $\epsilon$, takes
parameters $\eta$, $tb$, and $\vec{n}$;
$\T(\eta,tb,\vec{n})$ is an execution tree.
We write $\T, \lambda \models \psi$ if
$\T, (\eta, tb, \vec{n}), \lambda \models \psi$ for all inputs $(\eta,
tb, \vec{n})$.
Recall from Definition~\ref{def:comp-trace} that a protocol $\mathcal{Q}$ and an adversary $\mathcal{A}$ define a
model $\T_{\mathcal{Q},\mathcal{A}}$.
$\T_{\mathcal{Q},\mathcal{A}}(\eta,tb,\vec{n})$ is
the set of  traces of protocol $\mathcal{Q}$ generated
using
adversary $\mathcal{A}$, security parameter $\eta$, time bound $tb(\eta)$ on the
adversary's running time, and number of instances
of each protocol role given by $\vec{n} = \langle n_1,\ldots, n_k \rangle$.
For a set $\Delta$  of formulas, we define
$$ \Delta \Qmodels \psi \mbox{ iff }
\T_{\mathcal{Q},\mathcal{A}}, \lambda \models \Delta \mbox{ implies } \T_{\mathcal{Q},\mathcal{A}}, \lambda \models \psi  \mbox{ for all adversaries
$\mathcal{A}$, and \valuations $\lambda$}.$$

\noindent

\paragraph*{Definition of $\trace, \lambda \models \psi$}
We start by giving the semantics of basic formulas.
in the definition below, where the different elements of the trace are
given by
Definition~\ref{def:comp-trace}.

\begin{itemize}

\item

$\trace, \lambda \models {\pred{Send}(I, u)}$
iff $(\lamhat(I),\snd~\lamhat(u),\triv) \in \trace$.

\item

$\trace, \lambda \models \pred{Receive}(I,u)$
iff $(\lamhat(I),\receive,\lamhat(u)) \in \trace$.

\item

$\trace, \lambda \models \pred{Sign}(I, u, u')$
iff $(\lamhat(I),\action{sign}~\lamhat(u'), \lamhat(u)) \in \trace$.

\item

$ \trace, \lambda \models \pred{Verify}(I, w, u, \agent{B})$
iff $(\lamhat(I),\action{verify}~\langle\lamhat(w),\lamhat(u), vk({\lamhat}(\agent{B}))\rangle,\triv) \in \trace$.
\item

$\trace, \lambda \models \pred{New}(I,u)$ iff
$(\lamhat(I), \action{new},\lamhat(u)) \in \trace$.

\item

$\trace, \lambda \models {a_1 < a_2}$
iff $\trace, \lambda \models {a_1}$, $\trace, \lambda \models {a_2}$, and $\action{a_1}$
appears before $\action{a_2}$ in $\trace$, where $\action{a_1}$ and
$\action{a_2}$ are the actions in $\trace$ corresponding to
$a_1$ and $a_2$, respectively.
(Informally, this predicate models the temporal ordering on actions.)

\item

$\trace, \lambda \models {u=v}$ if $\lambda(u)=\lambda(v)$,
where $=$ is equality on values.

\item

$\trace, \lambda \models \Start(I)$ if no action label in $\trace$ has the form $(\lamhat(I), ., .)$.
(Intuitively, this predicate holds on traces in which the thread with identifier $I$ has executed no actions.)

\item

$\trace, \lambda\models {\pred{Contains}(u,v)}$ iff there exists
a series of identity function applications, projections, and
message-recovery operations from
signatures (deriving message from its signature) constructing $\lamhat(u)$ from $\lamhat(v)$.

\item

  $\trace, \lambda \models {\pred{Honest}(\agent{A})}$ iff $\agent{A}$ is a principal executing
  a protocol role on $\trace$.

\item

$\trace, \lambda \models {\theta\wedge\varphi}$ iff $\trace, \lambda \models \theta$ and $\trace, \lambda \models \varphi$.

\item

$ \trace, \lambda \models {\neg\varphi}$ iff $\trace, \lambda \not \models \varphi$.

\end{itemize}

We next define the semantics of modal
formulas of the form $\theta [P]_X \varphi$.
As we said above,
intuitively, $\theta [P]_I \varphi$ is true on a trace $\trace$ if, when $\trace$
is split into three pieces $\trace_1; \trace_2; \trace_3$,
if $\trace_2$ ``matches'' the program $[P]_I$, and if the precondition
$\theta$ holds for the trace $\trace_1$,
then the postcondition $\varphi$ holds for the trace $\trace_1;\trace_2$.
In order to make this intuitive idea precise, we define a notion of
\emph{matching} a program $[P]_I$ to a suffix of a trace. Variables in the
postcondition $\varphi$ can be bound to variables declared within $P$.
Matching $[P]_I$ to a trace may place constraints on how these
variables are interpreted by a \valuation.  For example, consider the formula
$true[x \leftarrow \irl{new}]_I\pred{New}(I, x)$, where $x$ is bound within the
program. On matching the program to a trace, a concrete value of $x$ for
the particular trace is obtained, which is used to evaluate the postcondition $\pred{New}(I, x)$.

\newcommand{\matches}[5]{#1 \gg #2 | #3, #4, #5}
We define what it means for a trace $\trace$ and \valuation
$\lambda$ to match program $[P]_I$ between prefixes $\trace_b$ and
$\trace_e$ and produce a \valuation $\lambda'$, denoted
$\matches{\trace, \lambda}{[P]_I}{\trace_b}{\trace_e}{\lambda'}$, by
induction on $[P]_I$:
\begin{itemize}
  \item $\matches{\trace, \lambda}{[l:x\leftarrow
    \alpha]_I}{\trace_1}{\trace_2}{\lambda[x \mapsto t]}$ if
    $\trace_2 \preceq \trace$ and $\trace_2$ ends in a transition
    $\ctrans{\conf}{c}{(\lamhat(I), \lamhat(\alpha), t)}{\conf'}$ for some configurations
    $\conf, \conf'$.
  \item $\matches{\trace, \lambda}{[l:x\leftarrow \alpha;P]_I}{\trace_1}{\trace_2}{\lambda'}$,
    if there exist $\trace_1'$, $\trace_2'$, and $\lambda_1$ such that
 (a) $\matches{\trace, \lambda}{[l:x\leftarrow
        \alpha]}{\trace_1}{\trace_1'}{\lambda_1}$ and (b)
    $\matches{\trace,
      \lambda_1}{[P]_I}{\trace_2'}{\trace_2}{\lambda'}$.
\end{itemize}

The semantics of the modal formula can now be defined as follows:
\begin{itemize}
\item $\trace, \lambda \models {\theta[P]_{I}\varphi}$
iff, for all $\trace_1$, $\trace_2$, and \valuations $\lambda'$, if
$\matches{\trace, \lambda}{[P]_I}{\trace_1}{\trace_2}{\lambda'}$ and
$\trace_1,\lambda \models \theta$, then $\trace_2,  \lambda'\models \phi$.
\end{itemize}

Finally, we give semantics to the universal quantifier in the standard way:
\begin{itemize}
\item
  $\trace, \lambda \models {\forall x.\,\varphi}$ iff $\trace, \lambda[x \mapsto v] \models \varphi$ holds
for all values $v$. Note that as with the protocol programming language,
the set of values is the set of all bitstrings and tuples of bitstrings.
\end{itemize}

\paragraph*{Definition of $\T, (\eta, tb, \vec{n}), \lambda \models \psi$}
We now give the semantics of $\theta \pimp{\epsilon} \varphi$,
where $\theta$ and $\varphi$ are basic formulas. In order to do so, we
associate with each basic formula $\phi$ a function $\Sem{\phi}$
from execution trees to execution trees.
Intuitively, for each execution tree $\tree$,
$\Sem{\phi}^\lambda(\tree)$ is the tree composed of those traces
 that satisfy
$\phi$ under \valuation $\lambda$.
Since all basic formulas $\phi$ are trace properties,
$\Sem{\phi}^\lambda(\tree) = \{\trace \in \tree | \; \trace,\lambda \models \phi\}$.
The formula $\theta \pimp{\epsilon} \varphi$ is intended to
express the fact that the probability
of $\varphi$ given $\theta$ is within $\epsilon$ of 1.
But since
$\epsilon$ is a function of $\eta,tb,\vec{n}$, the semantic relation is parameterized
by $\eta,tb,\vec{n}$.
\renewcommand{\T}{{\mathcal T}}
We have

\begin{itemize}
\item
$\T, (\eta, tb, \vec{n}), \lambda \models \varphi$  if $\Sem{\varphi}^\lambda{}{}(\T(\eta,tb,\vec{n})) = \T(\eta,tb,\vec{n}) $

\end{itemize}

\begin{itemize}
\item
$\T, (\eta, tb, \vec{n}), \lambda \models \theta \pimp{\epsilon} \varphi$  if
$$\frac{\measure{\,\Sem{\varphi}^\lambda{}{}(\T(\eta,tb,\vec{n})) \cap
\Sem{\theta}^\lambda{}{}(\T(\eta,tb,\vec{n}))\,}}{\measure{\,
\Sem{\theta}^\lambda{}{}(\T(\eta,tb,\vec{n}))\,}}\,\geq\:1-\epsilon(\eta,
tb, \vec{n}),$$
where we take the ratio to be 1 if
$\measure{\,\Sem{\theta}^\lambda{}{}(T(\eta,tb,\vec{n}))\,}= 0$.
\end{itemize}

Recall that $B^\epsilon \varphi$ is the special case of
$\theta \pimp{\epsilon} \varphi$
where $\theta = \true$. Also, it can be shown that $\measure{\,\T(\eta,tb,\vec{n})
\,} = 1$.  Thus, we have
\begin{itemize}
\item
$\T, (\eta, tb, \vec{n}),\lambda \models B^{\epsilon} \varphi$ if
$$\measure{\,\Sem{\varphi}^\lambda{}{}(\T(\eta,tb,\vec{n}))
\,}\,\geq\:1-\epsilon(\eta, tb, \vec{n}).$$
\end{itemize}

Boolean operations and universal quantification over conditional
formulas are defined the same way as for basic formulas.
\begin{itemize}
\item
  $\T, (\eta, tb, \vec{n}), \lambda \models \psi_1 \land \psi_2$ iff $\T,  (\eta, tb, \vec{n}), \lambda \models \psi_1$ and $\T,  (\eta, tb, \vec{n}), \lambda \models \psi_2$.
\item
  $\T, (\eta, tb, \vec{n}), \lambda \models \neg \psi $ iff $\T,  (\eta, tb, \vec{n}), \lambda \not\models \psi$.
\item
$\T, (\eta, tb, \vec{n}), \lambda \models {\forall x.\,\psi}$ iff $\T,  (\eta, tb, \vec{n}), \lambda[x \mapsto v] \models \psi$ holds
for all values $v$.
\end{itemize}

\paragraph{Example (Authentication Property)}
Informally, a matching conversations
property states that after an agent executes a particular role,
it is known that all the messages sent during the role were actually
received by the intended recipient, except possibly the last one.
This means that the role of the adversary is limited to that of a wire,
faithfully transferring messages to its intended recipients.

In the running example, a matching-conversations form of the authentication property
from the point of view of an initiator is formally expressed in QPCL as:
\begin{equation}\label{eq1}
\begin{array}{ll}
	B^{\epsilon}(
	 	\true [\cordname{Init}_{CR}(\agent{B})]_A\exists\iota. \\
			\quad \quad	(\Send(A,m) < \Receive(B,m) \land \\
                      \quad \quad	\Receive(B,m) < \Send(B,\<y, s\>) \land \\
			\quad \quad	\Send(B,\<y, s\>) < \Receive(A,\<y, s\>))).
\end{array}
\end{equation}
In formula (\ref{eq1}), $B$ is the thread identifier $\<\agent{B}, \iota\>$. Intuitively, the formula says that, with uncertainty $\epsilon$, after
$A$ executes the initiator's
$\cordname{Init}_{CR}$ role for the example challenge-response protocol, the messages
between $A$ and responder $B$ occur in order in all but a fraction $\epsilon$ of traces of the protocol.
Moreover, the proof of the protocol gives us a concrete bound on
$\epsilon$ that is roughly\footnote{The actual bound involves a few more terms, and is
presented in Section~\ref{sec:example}.}

$$\epsilon^\ufcma_{S}(tb, n_{init} + n_{resp}, \eta)+2\cdot 2^{-\eta} + 2\epsilon^\prg_P(tb, n_{init} + n_{resp},
\eta).$$ Here, $S$ is the signature scheme and $P$ is the pseudorandom-number
generator used in the protocol. Recall from Section~\ref{sec:primer} that
$\epsilon^\ufcma_{S}$ is the exact bound on the $\ufcma$ security
of the signature scheme $S$ and $\epsilon^\prg_{P}$ is the exact bound on the
$\prg$ security of the pseudorandom-number generator $P$. Also, $n_{init}$ and
$n_{resp}$ are the number of sessions of the initiator and the
responder roles, respectively.
The overall bound on the security of the protocol is therefore closely related
to the security of the underlying cryptographic primitives.
An outline of the proof of this formula is given in Section
\ref{sec:example}; a complete proof is provided in Appendix ~\ref{app:example}.

\newcommand{\qcr}{\mathcal{Q}_\mathit{CR}}
\newcommand{\initcr}{\mathbf{Init}_\mathit{CR}}
\newcommand{\respcr}{\mathbf{Resp}_\mathit{CR}}
\newcommand{\riff}{\Leftrightarrow}
\newcommand{\Verify}{\mathsf{Verify}}
\newcommand{\Decrypt}{\mathsf{Decrypt}}
\newcommand{\Has}{\mathsf{Has}}
\newcommand{\Gen}{\mathsf{Gen}}
\newenvironment{prog}{\begin{array}[t]{@{}l@{}}}{\end{array}}
\newenvironment{bprog}{\begin{array}[b]{@{}l@{}}}{\end{array}}

\subsection{QPCL Proof System}
\label{section:proof-system}

The proof system for QPCL includes axioms that capture properties of
cryptographic primitives (e.g., unforgeability of signatures), a rule for
proving invariants of protocol programs (useful for reasoning about properties
that capture the behavior of honest principals executing one or more instances
of protocol roles), as well as axioms and rules for first-order reasoning about
belief formulas.
All axioms are annotated with exact bounds.
Inference rules compute bounds for the formula in the consequent from
the bounds of
formulas in the antecedent (e.g., the conjunction of two belief formulas has a
bound that is the sum of the bounds for the two formulas).  Since a proof of a
security property can be visualized as a tree with axioms at the leaves and
rules at intermediate nodes, tracking bounds in this manner ensures that the
axiomatic proof yields an exact bound for the final property that is proved.

The exact bound in some of the axioms depends on the protocol $\mathcal{Q}$
whose correctness we are trying to prove. As a consequence, the entailment
relation $\Qentails$ is parameterized by the protocol $\Q$. Axioms that are
sound for any protocol are marked by the unparameterized entailment relation
$\entails$.

The exact bound $\epsilon$ is a function of the number of instances of the
protocol roles $\vec{n}$, the security parameter $\eta$, and the running time
of the adversary $tb(\eta)$.
Axioms related to cryptography introduce error bounds, while other proof rules
simply accumulate these bounds. As a result, error bounds of formulas provable
in QPCL are linear combinations of the error bounds introduced by the axioms
related to cryptography. In the present version of QPCL, error bounds have
one of the following forms:
\[\epsilon ::= 0 ~|~ \ever ~|~ \efstwo ~|~ \epsilon_1 + \epsilon_2.\]
Here, $\ever$ and $\efstwo$ are functions that represent the concrete bounds of
axioms \axname{VER} and \axname{FS2} discussed below, and $\epsilon_1 +
\epsilon_2$ is the function that is pointwise addition of $\epsilon_1$ and
$\epsilon_2$.

\renewcommand{\B}{B} The main technical result of
this section is that the proof system is \emph{sound}, that is,  given
protocol $\mathcal{Q}$, $\vec{n}$, $\eta$, and $tb$, each provable formula
$\B^{\epsilon} \varphi$ holds in each of the semantic models defined by
$\mathcal{Q}$, $\vec{n}$, $\eta$ and $tb(\eta)$ with probability at least $1 -
\epsilon(\vec{n}, \eta, tb)$.

The soundness proofs for axioms capturing properties of cryptographic
primitives are done via exact security reductions; they tightly relate
associated bounds to the bounds for the underlying primitive (e.g., digital
signature or pseudorandom-number generator). The exact security reduction
proofs proceed by showing that with high probability, whenever a violation of
an axiom is detected on a trace, with high probability, a violation of the
cryptographic primitive can be constructed. To detect a violation of an axiom
on a trace, however, additional computation needs to be performed, which
weakens the probability bound on the protocol.  Essentially, the reduction
proof says that an adversary can break an axiom in time $t$ if an adversary can
break the corresponding primitive $t+\delta$. We call $\delta$ the
\emph{reduction overhead}.
Specifying a concrete bound for $\delta$ requires implementation details of the
reduction such as the hardware on which the reduction is implemented.  We
elide these details in the paper. Appendix ~\ref{sec:soundness} contains the
reductions for axioms involving cryptography. The exact bound for the reduction
overhead for an axiom can be computed for any concrete implementation of the
reduction. We are satisfied when the overhead is a small polynomial in the
security parameters and number of sessions.

We summarize
below some representative axioms and proof rules (additional inference rules
and detailed soundness proofs are in Appendix ~\ref{sec:soundness}).

\paragraph*{Axioms about cryptographic primitives}
The axiom below captures the hardness of forging signatures.
The natural reading is that if thread $Y$ verified the signature of principal $\agent{B}$,
then some thread $\iota$ of that principal must have produced the signature.
There is one instance of this axiom for each fixed (constant) principal \agent{B}.
\begin{description}
\item[\axname{VER}.]
$
\QentailsC{\Delta,\pred{Honest}(\agent{B}) }{ \B^{\ever} (\pred{Verify} (Y, s, m, \agent{B})\\
~~~~~~~~~~~\Rightarrow \exists \iota. \pred{Sign} (\<\agent{B}, \iota\>, s, m) < \pred{Verify} (Y, s, m, \agent{B}))}.
$

\end{description}

Axiom \axname{VER} is not sound in general.  It is sound only if the
bound $\ever$ is chosen appropriately for the protocol $\Q$ (which is
why we use the notation $\Qentails$).
The appropriate choice of the bound is closely related to
$\epsilon^\ufcma_S$, the security bound for the signature scheme.
If $\Q$ involves the protocol roles $R_1, \ldots, R_k$, then
we define $\ever$ as follows:
\[\ever(\eta,tb,\vec{n}) = \epsilon^\ufcma_{S}(\eta, tb + \delta, q),\]
where $q = \Sigma_{i=1}^{k}\; (n_i \times q_i)$ is a bound on
the total number of signing actions on a trace, $n_i$ is the number of
instances of $R_i$ in $\Q$, $q_i$ is the number of signing actions in role $R_i$,
and $\delta$ is the reduction overhead for the axiom that can be bounded
by a polynomial that is $O(q_i^2\eta)$.\footnote{The concrete reduction
for \axname{VER} is given in Algorithm~\ref{fig:ver-instr} in Section~\ref{sec:soundness}. The overhead $\delta$ is
the concrete runtime for the reduction.}

The next axiom captures the hardness of predicting random nonces generated
using a pseudorandom-number generator (PRG). It states that if thread with identifier $X$
generated nonce $n$ and sent it out for the first time in message $m$, then if
$n$ is computable from a message $m'$ received by thread $Y$, then that receive
action must have occurred after the send action by thread $X$.

\begin{description}
\item[{\axname{FS2}}.]
  $\begin{array}{l}
\QDentails{\B^{\efstwo} (\pred{FirstSend} (X, n, m) \land \pred{Receive} (Y, m') \\
 \quad \quad \quad \quad \quad \land \pred{Contains} (n, m') \\
 \quad \quad \quad \quad \quad \quad \Rightarrow \pred{Send} (X, m) < \pred{Receive} (Y, m'))}.
\end{array}
$
\end{description}
FS2 is sound for a protocol $\mathcal{Q}$ if it is
\emph{nonce-preserving}, which  is the case if every program in $\Q$
satisfies the following two (syntactic) conditions:
\begin{enumerate}
  \item Fresh nonces (nonces not sent out on the network) should not be
    contained in the input to the $\irl{verify}$ action.
  \item Once a message containing a fresh nonce has been
    signed, no further messages are signed by the program until that
    message is sent on the     network.
\end{enumerate}

As with \axname{VER}, \axname{FS2} is sound only if the bound
$\efstwo$ is chosen appropriately.
The appropriate choice of  bound is closely related to
$\epsilon^\prg_P$, the security bound
of the pseudorandom-number generator $P$ specified by protocol $\Q$.
We define $\efstwo$  as follows:
\[\efstwo(\eta,tb,\vec{n}) = \epsilon^\prg_{P}(\eta, tb + \delta, q_n)
+ \eta\times q\times2^{-\eta}, \]
where $q_n$ is the total number of $\irl{new}$ actions on a trace, $q
= n/\eta*q_r$ is an upper bound on
the number of nonces that the adversary can send to a protocol thread
on a $\irl{receive}$,
$n$ is the length of the longest message that can be received,
$q_r$ is the number of receive actions in a trace, and $\delta$ is the reduction
overhead for the axiom and can be bounded by a  polynomial that is
$O(q_r\eta^2)$.\footnote{Similar to the above, the concrete reduction
for \axname{FS2} is  given in Algorithm ~\ref{fig:fs-instr} in Section~\ref{sec:soundness}.}

\paragraph*{Proof rule for protocol invariants}
To give the proof rules for invariants, we need a little notation.
$IS(\Q)$
(read as the \emph{initial segments} of $\Q$) denotes the set of all prefixes of
each of the roles in $\Q$. Since, a protocol is composed of only a finite number
of roles, the proof rule has only a finite number of antecedents.

\begin{description}
\item[\axname{HON}.]\ \
  $\begin{array}{c}\QDentails{\forall I.(\pred{Start}(I) \Rightarrow \varphi)} \\ \forall P \in IS(\mathcal{Q}).\quad  \QDentails{\forall I.(\pred{Start}(I)[P]_I\varphi)} \\ \hline {\QDentails{\varphi}} \end{array}$
\end{description}

\noindent Informally, this rule states that if a property $\varphi$ holds at the
beginning of a trace and is preserved by every protocol program, then $\varphi$
holds at the end of a trace.

\paragraph*{Proof rules for modal formulas}
We have the following rule and axiom for reasoning about modal
formulas:
\begin{description}
\item[\axname{PC}$_{\uparrow}$.] $\inferrule{ \Dentails \B^\epsilon
  \varphi }{ \Dentails \B^{\epsilon} (\theta [P]_I \varphi) }$
\item[\axname{PC}$_{\Rightarrow}$.] $\Dentails ((\theta
[P]_X (\varphi \Rightarrow \psi)) \, \wedge  \, (\theta [P]_X
\varphi)) \Rightarrow (\theta [P]_X \psi)$.
\end{description}

\medskip

The variant of \axname{PC}$_{\uparrow}$ that does not mention belief---from
$\phi$ infer $\theta [P]_I \phi$---is also sound and, indeed, can be
derived from \axname{PC}$_{\uparrow}$ and the rules for reasoning about
beliefs given below.  We seem to need the stronger version to prove results
about protocols of interest.  We remark that the axiom
$(\B^{\epsilon_1} (\theta
[P]_I (\varphi \Rightarrow \psi)) \wedge  \B^{\epsilon_2} (\theta [P]_I
\varphi)) \Rightarrow \B^{\epsilon_1 + \epsilon_2} (\theta [P]_I \psi)$,
which we actually use in our proof, can be
derived easily from \axname{PC}$_{\Rightarrow}$ and the axioms
for belief given below.

We have the following rules for Hoare-Style reasoning:

\begin{description}
\item[\axname{G1}.] $\inferrule{ \Dentails \theta[P]_X\varphi_1 \quad \Dentails\theta[P]_X\varphi_2 }{ \Dentails \theta[P]_X\varphi_1\land\varphi_2}$.
\item[\axname{G2}.] $\inferrule{ \Dentails \theta_1[P]_X\varphi \quad \Dentails\theta_2[P]_X\varphi }{ \Dentails \theta_1\lor\theta_2[P]_X\varphi}$.
\item[\axname{G3}.] $\inferrule{ \Dentails \theta\Rightarrow\theta' \quad \Dentails\theta'[P]_X\varphi' \quad \Dentails \varphi' \Rightarrow \varphi }{ \Dentails \theta[P]_X\varphi}$.
\end{description}

We also have the following rule for sequential composition:

\begin{description}
\item[\axname{S1}.] $\inferrule{ \Dentails \varphi_1[P]_X\varphi_2 \quad \Dentails\varphi_2[P']_X\varphi_3 }{ \Dentails \varphi_1[PP']_X\varphi_3}$.
\end{description}

\paragraph*{First-order reasoning about beliefs}
Our proof system includes a complete proof system for first-order belief logic
that allows us to prove formulas of the form $\B^{\epsilon} \varphi$ assuming a
set of formulas of the form $\B^{\epsilon'} \psi$.
The axioms and inference rules are just specializations of more
general axioms and rules given by Halpern \citeyear{Hal37}.
Here is (a slightly simplified version of) the proof system\footnote{The more general
versions of axioms B1-3 are presented in Appendix~\ref{section:conditionallogic}}:

\noindent \axname{B1}.
$\inferrule{\Dentails{\varphi \Rightarrow \phi'}}
{\Dentails{\B^{\epsilon} \varphi \Rightarrow \B^{\epsilon} \phi'}}$.

\noindent \axname{B2}. \ $\Dentails{ \B^0 (\true) }$

\noindent \axname{B3}. \ $\Dentails{\B^{\epsilon} \varphi \land \B^{\epsilon'} \varphi'
\Rightarrow B^{\epsilon + \epsilon'} (\varphi \land \varphi')}$

\smallskip

The next axiom is sound in our setting, where all traces get positive
probability.  This means if $\phi$ holds with probability 1, then it
must hold on all traces, since no trace has probability 0.
\smallskip

\noindent \axname{B4}. \ $\Dentails{\B^0 (\phi) \Rightarrow \phi}$.

\medskip

We include all first order tautologies as axioms in our proof system.

\subsection{Soundness}
\label{sec:soundness}

We now state our soundness theorem.
$\Qentails \psi$
 denotes that $\psi$
is provable using instances of the axiom and
proof rule schemas for protocol $\mathcal{Q}$ (that is, all rules of
the form $\Qentails \ldots$ or $\entails \ldots$).

\begin{theorem} (Soundness)
If $\Qentails{\psi}$
then $\Delta\Qmodels  \psi$.
\end{theorem}

We first present definitions and lemmas required
to prove the soundness of QPCL's proofsystem,
and then prove the soundness of the axioms about
cryptographic primitives as illustrative cases. Appendix~\ref{sec:soundness-app} contains proofs of other axioms.

Each property that is provable using QPCL's proof system is
a safety property in the following sense.

\begin{definition}
  A basic formula $\varphi$ is called a \emph{safety property} if
  $\trace, \lambda \not\models \varphi$ implies that for all traces
  $\trace'$ such that $\trace \preceq \trace'$, $\trace', \lambda$ and $
  \not\models \varphi$.
\end{definition}

Safety ensures that properties which are false on a trace continue to be false.
This allows us the bound the probability of a formula being false in the
middle of an execution, by the probability at the end of the execution.

Since safety applies only to trace properties and not to formulas with belief
operators, we convert probabilistic assertions to basic formulas using an
erasure operation $\erase{\cdot}$.

\begin{definition}[Erasure]
  The erasure operation $\erase{\psi}$ is defined by induction on $\psi$ as follows:
      \begin{itemize}
        \item $\erase{\B^\epsilon(\varphi)} = \varphi$.
        \item $\erase{\psi_1 \land \psi_2} = \erase{\psi_1} \land \erase{\psi_2}$.
        \item $\erase{\neg\psi} = \neg\erase{\psi}$.
        \item $\erase{\forall x.\psi} = \forall x.\erase{\psi}$.
        \item $\erase{\varphi} = \varphi$.
      \end{itemize}
\end{definition}

\begin{lemma}[Safety]
\label{lemma:safety}
If $\Qentails \psi$, then $\erase{\psi}$ is a safety property.
\end{lemma}

\begin{proof}
By induction on the derivation of $\Qentails \psi$. We consider three
cases here, and leave the remaining cases to the reader:
\begin{description}
  \item[\axname{B1}.] Assume by the induction hypothesis that $\varphi
    \rimp \varphi'$ is a safety property.
    Then it is immediate that $\erase{\B^\epsilon(\varphi) \rimp
      \B^\epsilon(\varphi')}$ is
a safety property.

  \item[\axname{PC}$_{\uparrow}$.] Assuming, by the induction hypothesis,
    that $\varphi$ is a safety property, we must
    show that $ \theta[P]_I\varphi$ is a safety property. Assume
    that for some $\trace$, $\lambda$, it is the case that
    $\trace, \lambda \not\models \theta[P]_I\varphi$. By the
    semantics of the modal formula,
    we know that there exists $\tra_1$, $\tra_2$, $\lambda'$ such that $(\matches{\trace, \lambda}{[P]_I}{\tra_1}{\tra_2}{\lambda'})$,
    $\trace_1, \lambda \models \theta$, and $\tra_2, \lambda' \not\models \varphi$. For
    all traces $\trace'$ such that $\trace \preceq \trace'$, we must have
    that $\matches{\trace',
      \lambda}{[P]_I}{\trace_1}{\trace_2}{\lambda'}$. Therefore,
    $\trace', \lambda \not\models \theta[P]_I\varphi$.

  \item[\axname{FS2}.] We need to show that $\varphi_{\axname{FS2}} =
    \pred{FirstSend}(X, n, m) \land \pred{Receive}(Y, m') \rimp
    \pred{Send}(X, m) < \pred{Receive}(Y, m')$ is a safety
    property. Assume that for some trace $\trace$, $\lambda$,
    it is the case that $\trace, \lambda \not\models \varphi_{\axname{FS2}}$.
    Therefore, it must be the case that $(\lamhat(Y), \irl{receive},
    \lamhat(m')) \in \trace$ and that $(\lamhat(X), \irl{sent}~\lamhat(m),
    \triv)$ does not appear in $\trace$ before $(\lamhat(Y), \irl{receive},
   \lamhat(m')) \in \trace$. Therefore, the same must be true for all
    $\trace'$ such that $\trace \preceq \trace'$.
\end{description}
\end{proof}

\paragraph*{Soundness of the axioms about cryptographic primitives}
We prove axioms about cryptographic primitives sound by reducing the validity
of the axiom to the security of the primitive. Such proofs require augmenting
the implementations of primitives with a program that monitors the validity of
the axiom on the trace. When a violation of the axiom is detected, it can be
transformed into a violation of the primitive's security.  As every attack on
the axiom can be transformed in to an attack on the primitive, the probability
of the primitive being broken is a bound on the probability that the axiom is
false. However, we need to account for the additional time required to monitor
the validity of the axiom while computing the probability of attack on the
primitive.

In order to prove the soundness of \axname{VER}, we need to provide a
formal definition of execution tree, so that we can do an induction on
the structure of the execution tree.

We define a prefix relation on trees. Intuitively, a tree $\tree$
is a prefix of $\tree'$ if every trace in $\tree$ is a prefix of $\tree'$.

\begin{definition}[Tree Prefix]
A tree $\tree$ is a \emph{prefix} of $\tree'$, denoted $\tree \preceq \tree'$, if
one of the following conditions hold:
\begin{enumerate}
  \item $\tree = \tree'$,
  \item $\tree$ is a configuration $\conf$ and $\tree' = (\conf \rightarrow \tree_1 ~|~\cdots~|~\conf \rightarrow \tree_k)$\footnote{
   Cost and probability labels are elided here for brevity. },
  \item $\tree = (\conf \rightarrow \tree_1 ~|~\cdots~|~\conf
    \rightarrow \tree_k)$,
$\tree' = (\conf \rightarrow \tree'_1 ~|~\cdots~|~\conf \rightarrow \tree'_k)$,
    and $\tree_1 \preceq \tree'_1 \cdots \tree_k \preceq \tree'_k$.
\end{enumerate}
\end{definition}

We now present an equivalent definition of the probability of a
formula interpreted on a tree. The function $\Pr(\varphi, \trace, \tree, \lambda)$
denotes the probability of a formula $\varphi$ in the subtree
$\tree$ that follows after trace $\trace$, with \valuations $\lambda$. Figure
$~\ref{fig:prob}$ shows a trace $\trace$ followed by a subtree $\B_{\trace}$
of $\tree$. Since QPCL formulas depend on an entire trace, it is not enough to
define the probability of a formula with respect to a subtree, but to include
the history of events that ocurred before the subtree as illustrated in
Figure~\ref{fig:prob}. In the definitions below, $\trace \rightarrow \conf$
denotes the the trace obtained by extending $\trace$ with the transition
$\conf' \rightarrow \conf$, where $\conf'$ is the last configuration in $\trace$.
The empty trace is denoted by $\cdot$, and $\cdot \rightarrow \conf$ is the
trace $\conf$, containing only a single configuration and no transitions.

\begin{figure}
  \centering
		\includegraphics[width=0.4\textwidth]{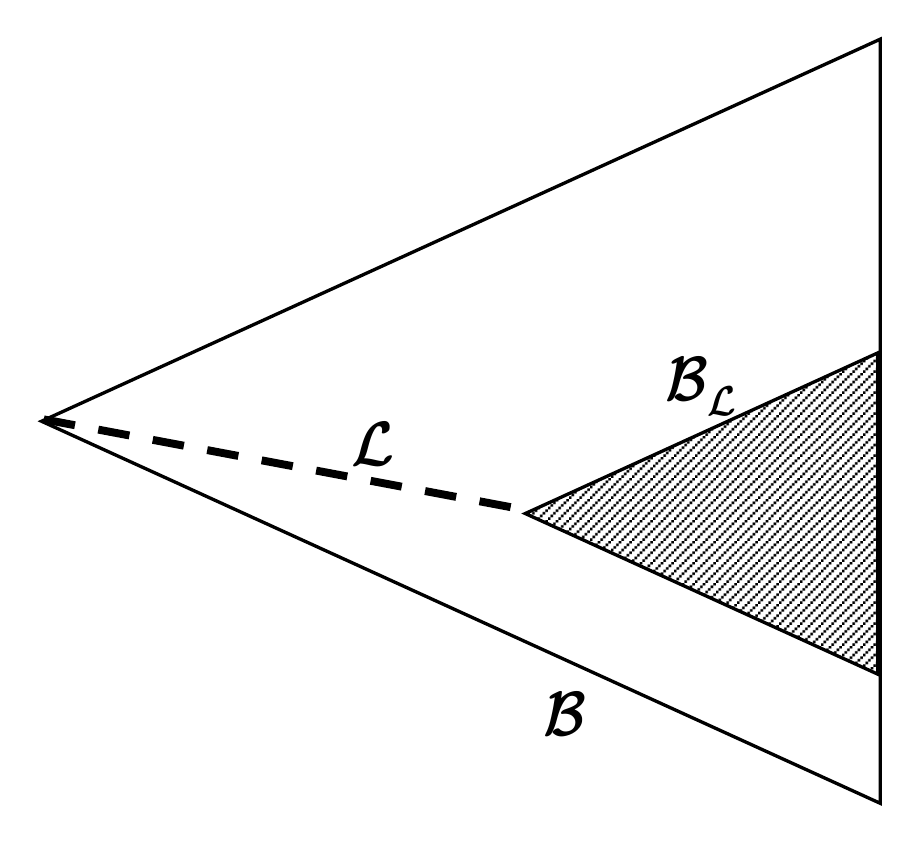}
        \caption{An Execution Tree with a trace $\tra$ and subtree $\tree_\tra$}
   \label{fig:prob}
 \end{figure}

\begin{definition}[Probability]
$\Pr(\varphi, \trace, \tree, \lambda)$ is defined inductively on $\tree$
as follows:
\begin{enumerate}
\item $\Pr(\varphi, \trace, \conf, \lambda) = 1$ if $\trace \rightarrow \conf,\lambda \models \varphi$,
\item $\Pr(\varphi, \trace, \conf, \lambda) = 0$ if $\trace\rightarrow \conf ,\lambda \not\models \varphi$,
\item $\Pr(\varphi, \trace, \tree, \lambda) = \sum_{i=0}^k p_k \Pr(\varphi, \trace\rightarrow \conf, \tree_i, \lambda)$
   if $\tree = (\conf \xrightarrow{p_1, c_1} \tree_1 ~|~\cdots~|~\conf
   \xrightarrow{p_k,c_k} \tree_k)$.
\end{enumerate}
\end{definition}

The probability of a formula $\varphi$ interpreted on a tree $\tree$
preceded by
 the empty trace, and using
\valuation $\lambda$ is $\Pr(\varphi, \cdot, \tree, \lambda)$.
The following lemma states
 that the two definitions of probability are equivalent.
\begin{lemma}
\label{lemma:ind-equiv}
For all \valuations $\lambda$,
$\measure{[[\varphi]]^\lambda(\tree)} = \Pr(\varphi, \cdot, \tree, \lambda)$.
\end{lemma}

We also state a lemma that relates the probability of a safety
property on a tree $\tree$ to its probability on a prefix of $\tree$.
\begin{lemma}
\label{lemma:safety-prefix}
If $\varphi$ is a safety property and $\tree \preceq \tree'$,  then
for all $\lambda$,
$\measure{[[\varphi]]^\lambda(\tree)} \geq \measure{[[\varphi]]^\lambda(\tree')}$.
\end{lemma}

In our proofs of the soundness of the axioms about cryptographic
primitives, we replace the implementations of some actions with
different implementations that preserve input-output behavior, in a
sense made precise by the following definition.

\begin{definition}[Simulation]
An implementation $\bar{A}$ \emph{simulates} an implementation $A$
with the relation $\sim$ on $S_A \times S_{\bar{A}}$
(where $S_B$ denotes the set of states of implementation $B$), if
(a) $\sigma_A^{\mathit{init}}\sim\sigma_{\bar{A}}^{\mathit{init}}$, where $\sigma_A^{\mathit{init}}$ and $\sigma_{\bar{A}}^{\mathit{init}}$
are the initial states of $A$ and $\bar{A}$ respectively and
(b) if $\sigma_A \sim \sigma_{\bar{A}}$, then for all
transitions $\sigma_A, A(t) \xmapsto[p]{c} \sigma_A',ra$ of $A$, there exists a
transition $\sigma_{\bar{A}},\bar{A}(t) \xmapsto[p]{c'}
\sigma_{\bar{A}}', ra$ of $\bar{A}$ such that $\sigma_A' \sim \sigma_{\bar{A}}'$.
If there exists a $b$ such that, for all transitions of $A$, if
$c$ is the cost of the transition
and $c'$ is the cost of the
corresponding transitions of $\bar{A}$, we have
$c \leq c' \leq c+b$,
 then $\bar{A}$ simulates $A$ \emph{with time bound $b$}.
\end{definition}
\noindent Intuitively, if $\bar{A}$ simulates $A$ with the relation $\sim$, then starting from
$\sim$-related states, the implementations $A$ and $\bar{A}$ exhibit the same
input-output behavior, but have possibly different costs.

We wish to lift simulations over action implementations to simulations over
setups, which are collections of action implementations.  However, for a setup
$\mathcal{S'}$ to simulate $\mathcal{S}$, it is not sufficient for
the corresponding implementations in $\mathcal{S}'$ and $\mathcal{S}$ to be
simulations.  Recall that different implementations may share state, and that
state sharing is governed by the label $\sttag{a}_I$ of an action $a$ in the thread
with identifier
$I$ in the global state. Therefore, we require implementations that share state
to respect each other's simulation relations.

\begin{definition}[Setup Simulation] A setup $\mathcal{S}'$ simulates a setup
$\mathcal{S}$ if there exists a set $\{\sim_{h_1}, \cdots,
\sim_{h_n}\}$ of relations, indexed by global state labels  $\{h_1,
\cdots, h_n\}$,  and each
$\impl{a}_I$ in $\mathcal{S'}$ simulates the corresponding implementation in
$\mathcal{S}$ with relation $\sim_h$, where $h = \sttag{a}_I$.
\end{definition}

We now show that if two setups $\mathcal{S}$ and $\mathcal{S}'$ differ only in
their action implementations, and the action implementations in $\mathcal{S}'$
simulate the corresponding implementations in $\mathcal{S}$, then for any safety property $\varphi$,
and assignment $\lambda$, the probability of $\varphi$ in the tree
generated by $\mathcal{S}$ is bounded by the probability of $\varphi$ in the
tree generated by $\mathcal{S}'$.

\begin{lemma}\label{lemma:simulation-bound}
If setup $\mathcal{S'}$ simulates
  setup $\mathcal{S}$, each implementation $\impl{a}_I$ in
  $\mathcal{S}'$ simulates the corresponding implementation in
  $\mathcal{S}$ with time bound $b_{a_I}$, and
  $\varphi$ is a safety property,
  then for all $\lambda$, $\measure{[[\varphi]]^\lambda(\mathcal{T}{(\mathcal{S}, \mathbf{t},
      \vec{n}, tb,
\eta)}}) \geq \measure{[[\varphi]]^\lambda(\mathcal{T}{(\mathcal{S'}, \mathbf{t},
      \vec{n},
tb + \Sigma_{a_I} n_{a_I} b_{a_I}, \eta)}})$, where $n_{a_I}$ is the
  number of times
implementation $a_I$ is called.
\end{lemma}

\begin{proof}
  We first show by induction on  $\mathcal{T}{(\mathcal{S}, \mathbf{t}, \vec{n}, tb,
\eta)}$ that $\mathcal{T}{(\mathcal{S}, \mathbf{t}, \vec{n}, tb,
\eta)} \preceq \mathcal{T}{(\mathcal{S'}, \mathbf{t}, \vec{n},
tb + \Sigma_{a_I} n_{a_I} b_{a_I}, \eta)}$. Essentially, we can show
that
each trace using the implementations in $\mathcal{S'}$ makes
transitions identical to the
  trace using the implementations in $\mathcal{S}$, and
takes no more time than the extra time allotted ($\Sigma_{a_I} n_{a_I}
b_{a_I}$).
We then apply Lemma~\ref{lemma:safety-prefix}.
\end{proof}

In the theorem above, we call the extra time allocated for
the simulation to execute
($\Sigma_{a_I} n_{a_I} b_{a_I}$) the \emph{reduction overhead}.

\begin{lemma} The inference rule \axname{VER} is sound. \end{lemma}

\begin{proof}
Let $\varphi_{\axname{VER}}$ abbreviate the formula $
  \pred{Verify} (Y, s, m, \agent{B})
 \Rightarrow \exists \iota. \pred{Sign} (\<\agent{B}, \iota\>, s, m)$.
We must show that  \mbox{$\Qmodels \B^{\ever} \varphi$}, where $\ever$
  is defined as
follows:
\[\ever(\eta,tb,\vec{n}) = \epsilon^\ufcma_{S}(\eta, tb + \delta, q),\]
where, $\epsilon^{\ufcma}_S$ is the concrete security bound of the signature
scheme $S$ specified by the protocol $\Q$, $q = \Sigma_{i=1}^{k}\; (n_i \times
q_i)$ is a bound on the total number of signing actions on a trace, $n_i$ is
the number of instances of $R_i$ in $\Q$, $q_i$ is the number of signing
actions in role $R_i$, and $\delta$ is the reduction overhead
for the simulation where the implementations of actions $\irl{sign}$ and
$\irl{verify}$ are replaced respectively by the implementations $\irl{sign}'$ and $\irl{verify}'$
described in Algorithm ~\ref{fig:ver-instr}.

  \begin{algorithm2e}
    \SetKwFunction{isign}{sign'}\SetKwFunction{iverify}{verify'}
    \SetKwFunction{sign}{sign}\SetKwFunction{verify}{verify}
    \SetKwProg{myalg}{Implementation}{}{}
    \myalg{\isign{$(\sigma, M), (m, k)$}}{
      \nl $(\sigma', r) \leftarrow$ \sign{$\sigma, (m,k)$}\;
      \nl \KwRet{$(\sigma', M\cup\{m\}), r$}\;}{}
    \setcounter{AlgoLine}{0}
    \SetKwProg{myproc}{Procedure}{}{}
    \myalg{\iverify{$(\sigma, M), (m, s, k)$}}{
      \nl $(\sigma', r) \leftarrow$ \verify{$\sigma, (m, s, k)$}\;
      \nl$bad \leftarrow (m \not\in M)\land(r \not= \abort) $\;
      \nl \KwRet{$(\sigma', M), r$}\;}{}
  \caption{Implementations $\mathtt{sign'}$ and $\mathtt{verify'}$ augment $\mathtt{sign}$ and $\mathtt{verify}$ to monitor \axname{VER}}
  \label{fig:ver-instr}
  \end{algorithm2e}

Given $\A, \eta, tb$, $\vec{n}$, and $\lambda$, we want to show that
\[\measure{\Sem{\neg\varphi_{\axname{VER}}}^\lambda(\T_{\Q,\A}(\vec{n}, tb, \eta))}
\leq \ever(\eta, tb, \vec{n}).\]
As a first step, we replace the implementations of $\irl{sign}$ and
$\irl{verify}$ in $\Q$ with
the implementations shown in Algorithm~\ref{fig:ver-instr}.
The implementation $\irl{sign}'$ keeps track of the messages signed;
$\irl{verify}'$ checks whether a verified message has been signed before. The
variable $bad$ in $\irl{verify}'$ is set to true when a message is verified
that has not been signed before. Therefore, $bad$ is set to true precisely when
$\neg{\varphi_{\axname{VER}}}$ is true on a trace.  Note that the
implementations $\irl{sign}'$ and $\irl{verify}'$ simulate $\irl{sign}$ and
$\irl{verify}$. The simulation relation is $\sigma \sim (\sigma, M)$ for all
$\sigma$ and $M$, and the initial state for $\irl{sign}'$ and $\irl{verify}'$
is $(\sigma^{init}, \emptyset)$, where $\sigma^\mathit{init}$ is the initial state for the
implementation of $\irl{sign}$ and $\irl{verify}$. Therefore, by
Lemma~\ref{lemma:simulation-bound}, we have that
\[\measure{\Sem{\neg\varphi_{\axname{VER}}}^\lambda(\T_{\Q,\A}(\vec{n}, tb, \eta))}
\leq \measure{\Sem{\neg\varphi_{\axname{VER}}}^\lambda(\T_{\Q',\A}(\vec{n}, tb +
  \delta, \eta))}.\]
Let $\Q'$ be the protocol that results by replacing the implementations
$\irl{sign}$ and $\irl{verify}$ in $\Q$ with $\irl{sign}'$ and $\irl{verify}'$.
Let $\delta$ be the additional runtime required
for the reduction in Algorithm~\ref{fig:ver-instr}; $\delta$ can be bounded by a
polynomial that is $O((q_i+q_j)^2\eta)$, where $q_i$ and $q_j$ are the number of signature and verification
actions in a role, assuming some quadratic implementation for the set operations
in Algorithm~\ref{fig:ver-instr}.
Recall from Section~\ref{sec:primer} that the $\ufcma$ security condition
is violated when the adversary is able to create a signature of
a message that has not been signed before. Therefore, when $bad$ is set
to true, the $\ufcma$ security of the scheme $S$ has been violated. Consequently, the
probability that $bad$ is set to true can be bounded by $\epsilon^{\ufcma}_S(tb+\delta, q, \eta)$,
where $q$ is the number of signing actions and is bounded by the number of signing
operations possible in running the protocol, that is, $q =
\Sigma_{i=1}^{k}\; (n_i
\times q_i)$ and $q_i$ is the number of signing actions in role $\rho_i$ of
protocol $Q = {r_1, \ldots, r_k}$. Since $bad$ is set to true when $\neg\varphi_{\axname{VER}}$
is true on a trace, we can show that
\[\measure{\Sem{\neg\varphi_{\axname{VER}}}^\lambda(\T_{\Q',\A}(\vec{n}, tb + \delta, \eta))} \leq
\epsilon^{\ufcma}_S(tb+\delta, q, \eta) = \ever(\eta, tb, \vec{n}).\]
\end{proof}

\newcommand{\nonces}[1]{\mathtt{nonces}(#1)}
\begin{lemma} The axiom \axname{FS2}
  is sound. \end{lemma}

\begin{proof}
Let $\afstwo$ abbreviate
$\pred{FirstSend} (X, n, m) \land
 \pred{Receive} (Y, m')
\land \pred{Contains} (n, m')
\Rightarrow \pred{Send} (X, m) < \pred{Receive} (Y, m').
$
We must show that $\Qmodels \B^{\efstwo} \afstwo$, where $\efstwo$ is defined as follows:
\[\efstwo(\eta,tb,\vec{n}) = \epsilon^\prg_{P}(\eta, tb + \delta, q_n)
+ \eta\times q\times2^{-\eta}, \]
where $q_n$ is the total number of $\irl{new}$ actions on a trace, $q =
l/\eta*q_r$ is an upper bound on the number of nonces that the adversary can
send to a protocol thread on a $\irl{receive}$, $l$ is the length of the
longest message that can be received, $q_r$ is the number of new, send, and receive actions in
a trace, and $\delta$ is the reduction overhead for the simulation where the
implementations of actions $\irl{new}$, $\irl{send}$ and $\irl{receive}$ are
replaced respectively by the implementations $\irl{new}'$,  $\irl{send}'$ and
$\irl{receive}'$ described in Algorithm~\ref{fig:fs-instr}.

\begin{algorithm2e}
    \SetKwFunction{isend}{send'}\SetKwFunction{irecv}{recv'}
    \SetKwFunction{send}{send}\SetKwFunction{recv}{recv}
    \SetKwFunction{new}{new}\SetKwFunction{inew}{new'}
    \SetKwProg{myalg}{Implementation}{}{}
    \myalg{\inew{$(\sigma, M), ()$}}{
      \nl $(\sigma', \ret{n}) \leftarrow$ \new{$\sigma, ()$}\;
      \nl \KwRet{$(\sigma', M \cup \{n\}), \ret{n}$}\;}{}
    \setcounter{AlgoLine}{0}
    \myalg{\isend{$(\sigma, M), (m, k)$}}{
      \nl $(\sigma', r) \leftarrow$ \send{$\sigma, (m,k)$}\;
      \nl \KwRet{$(\sigma', M - \nonces{m}), r$}\;}{}
    \setcounter{AlgoLine}{0}
    \SetKwProg{myproc}{Procedure}{}{}
    \myalg{\irecv{$(\sigma, M), ()$}}{
      \nl $(\sigma', r) \leftarrow$ \recv{$\sigma, ()$}\;
      \nl \lIf{$r = \ret{m}$}{$bad \leftarrow \exists n \in \nonces{m} \land n \in M $}
      \nl \KwRet{$(\sigma', M), r$}\;}{}
      \caption{Implementations $\mathtt{new'}$, $\mathtt{send'}$, and $\mathtt{recv'}$ augment $\mathtt{new}$, $\mathtt{send}$ and $\mathtt{recv}$ repectively to monitor \axname{FS2}}
  \label{fig:fs-instr}
  \end{algorithm2e}

Given $\A, \eta, tb$, and $\vec{n}$, we must show that
\[\measure{\Sem{\neg\afstwo}^\lambda(\T_{\Q,\A}(\vec{n}, tb, \eta))} \leq
\efstwo(\eta, tb, \vec{n}).\]
Let $\Q'$ be the protocol that results from replacing
the implementations of $\irl{new}$, $\irl{send}$, and
$\irl{receive}$
with the implementations shown in Algorithm~\ref{fig:fs-instr}.
The
set $M$ records the
nonces that have been generated but not been sent out. The implementation
$\irl{new}'$ records the nonces that have been generated. When
a message is sent with a nonce using $\irl{send}'$, it is
removed from $M$.
The variable $bad$ is set to true when there is a nonce in $M$ that has been received.
Essentially, $bad$ is set to true on a trace exactly when $\afstwo$ is false.

Consider the
following two experiments:
\begin{itemize}
\item $\textbf{Exp}^\text{Random}$: Run protocol $\Q'$, with $\irl{new}$ implemented
by a random-number generator with uniformly random distibution. If
$bad$ is set to true output $1$ else $0$.
\item $\textbf{Exp}^P$: Run protocol $\Q'$, with $\irl{new}$ implemented
by the pseudorandom-number generator $P$ specified by the
protocol. If $bad$ is set to true output $1$ else
$0$.
\end{itemize}

Recall from Section~\ref{sec:primer} that the probability of an adversary being
able to distinguish a uniform random number generator from the output of a
pseudorandom-number generator $S$ after seeing $q$ samples and in
time $t$ is bounded by
$\epsilon_P^\prg(\eta, t, q)$.
Therefore, we have the following:
\[\textbf{Pr}(\textbf{Exp}^P = 1) - \textbf{Pr}(\textbf{Exp}^\text{Random} = 1)  \leq
\epsilon^\prg_P(\eta, tb+\delta, q),\]
where $q_n = \Sigma_{i=1}^{k}\; (n_i \times q_i)$, where $q_i$ is the number of $\irl{new}$ actions in
role $i$ and $n_i$ is the number of instances of role $r_i$ in protocol roles $\{r_1, \cdots, r_k\}$. The additional runtime $\delta$
can be bounded by a polynomial that is $O(q_r^2\eta)$, where $q_r$ is
the total number of
new, send, and receive actions in the protocol, assuming some quadratic implementation
of the set operations in Algorithm~\ref{fig:fs-instr}.
The total runtime for $\epsilon^\prg$ is $tb+\delta$ to account for the additional
runtime for the reduction.

As $bad$ is set to true exactly when $\neg \afstwo$ is
true, as with the proof for $\axname{VER}$, we can use Lemma~\ref{lemma:simulation-bound}
to get the following bound for $\afstwo$:
\[ \measure{\Sem{\neg\afstwo}^\lambda\T_{\Q,\A}(\vec{n}, tb, \eta)} \leq  \textbf{Pr}(\textbf{Exp}^P(A) = 1). \]

We now only need to derive a bound for
$\textbf{Pr}(\textbf{Exp}^\text{Random} = 1)$, since
$\textbf{Pr}(\textbf{Exp}^P = 1) \leq \textbf{Pr}(\textbf{Exp}^\text{Random}(A) = 1) + \epsilon^\prg(\eta, tb+\delta, q)$.
When we use a true random number generator, at each guess, the adversary can
guess the random number generated with probability only $2^{-\eta}$. Therefore,
if the adversary has $q$ guesses, by the union bound, we can show that the total
probability of $bad$ being true is $q 2^{-\eta}$.  The adversary can attempt to
guess the nonce at each $\irl{receive}$ action. However, each message received
can contain multiple nonces. If the longest message in the protocol is of
length $l$, then each message received can contain at most $l/\eta$ nonces of
length $\eta$, and so a bound on $q$, the total number of guesses available to
the adversary, is $q_r l/\eta$. Finally, we have that
\[
\textbf{Pr} (\textbf{Exp}^\text{Random} = 1) \leq \frac{q_r l}{\eta2^{\eta}}.
\label{eq:rand-bound}
\]
Therefore, we have that
\[\measure{\Sem{\neg\afstwo}^\lambda\T_{\Q,\A}(\vec{n}, tb, \eta)} \leq
\textbf{Pr}(\textbf{Exp}^P = 1) \leq
\epsilon^\prg(\eta, tb+\delta, q) + \frac{q_r l}{\eta2^{\eta}} = \efstwo(\eta, tb, \vec{n}).\]
\end{proof}

\section{Proof Outline for the Example Protocol}
\label{sec:example}

A formal proof of the weak authentication property for the initiator
guaranteed by executing the CR protocol is presented in Appendix ~\ref{app:example}.
We explain below the outline of the proof and discuss how the axioms and rules described
in the previous section are used. The formal proof proceeds roughly as follows.

First, we assert what actions were executed by the thread with identifier $A$ in the initiator role. Specifically,
in this part of the proof, we prove that $A$ has received and
verified \agent{B}'s signature. We then use the fact that the
signatures of honest parties are unforgeable (axiom \axname{VER}) and
first-order reasoning about beliefs
to conclude that some thread with name $\iota$ of $\hat{B}$ ($B = \<\hat{B}, \iota\>$ below) must have produced this signature.
This step introduces a bound of $\ever$, and we obtain
the following formula:
\begin{equation}
  \QentailsC{\pred{Honest}(\agent{B})}{\B^{\ever}(
		 \true [\cordname{Init}_{CR}(\agent{B})]_A \ \pred{Sign} (B,
   s, \<\resp, y, m, \agent{A}\>))}.
	\label{eq:ysigned}
\end{equation}

Second, we use the honesty rule \axname{HON} to infer that whenever any thread (and by instantiation using the
rules of first-order belief logic, any the thread with identifier $B$ of agent \agent{B})
generates a signature of this form, he must
have previously received the first message from \agent{A} and then
sent the second message; moreover, the second message
was in fact the first send involving the random nonce $y$:

\begin{equation}
\label{eqn:inv}
  \begin{array}{ll} \Qentails{} &
		\B^0(\pred{Sign} (B, s, \<\resp, y, m,
\agent{A}\>) \Rightarrow (\Receive(B, m) \\
&~~~~~\land (\Send(B,\<y, s\>) \Rightarrow\\
&~~~~~~~~~~~~~ \Receive(B,m) < \Send(B,\<y, s\>)) \\
&~~~~~\land \pred{FirstSend} (B,\<y, s\>)
			)).
 		\end{array}
\end{equation}

Combining (\ref{eq:ysigned}) and (\ref{eqn:inv})
using rules for first-order reasoning
about beliefs, we conclude that the following
formula, which orders the actions performed by $B$, holds:

\begin{equation}
\label{eqn:msgorder1}
\begin{array}{ll} \Qentails{} & 
	\B^{\ever}(
		\true [\cordname{Init}_{CR}(\agent{B})]_A \\
    &~~~~~~~~~~~\Receive(B,m) < \Send(B,\<y, s\>)).
 	\end{array}
\end{equation}

Next we want to prove that $B$ received the first message after $A$ sent
it. We use the fact that since the thread with identifier $A$ generated the random nonce
$m$ and sent it out for the first time in the first message of the protocol,
then if the thread with identifier $B$ received a message from which $m$ is computable, he must have done so after $A$ sent out the first message.
This step uses the \axname{FS2} axiom, and therefore the bound
associated with the underlying PRG applies:

  \begin{equation}
\label{eqn:FS2}
  \begin{array}{ll}
\Qentails{} &
	\B^{\efstwo}  (		\true [\cordname{Init}_{CR}(\agent{B})]_A \\
  &~~~~	((\pred{FirstSend} (A, m,  m)  \land 			\Receive(B,m) \land \pred{Contains} (m, m)) \\
   &~~~~~~~~			\Rightarrow \Send(A,m) < \Receive(B,m)
		 )
	  ).
\end{array}
\end{equation}

Using the fact that $A$'s first send had the form in the premise of the implication in (\ref{eqn:FS2}) and that $B$ received the
corresponding message from (\ref{eqn:inv}),
from B3, (\ref{eqn:msgorder1}), and (\ref{eqn:FS2}), we infer
desired ordering of actions by $A$ and $B$.
Note that since two belief formulas are combined, the bounds add up:
$$\begin{array}{l}
  \pred{Honest}(\agent{B}) \entails_\Q \B^{\ever+\efstwo} (\true [\cordname{Init}_{CR}(\hat{B})]_A\\
~~~~~~~~~~~~~~~~~~~~~~~~ \ \Send(A,m) < \Receive(B,m)).
\end{array}
$$

The rest of the proof uses similar ideas (including another application
of \axname{FS2}) to order
other actions performed by threads $A$ and $B$ to prove a matching-conversation like property of
the protocol with a bound that is tightly related to the security of the underlying digital signature scheme
and pseudorandom-number generator:
$$
\begin{array}{ll}
  \pred{Honest}(\agent{B})\vdash_\Q	& \B^{\ever + 2\cdot\efstwo}(
  \true [\cordname{Init}_{CR}(\agent{B})]_A \ \exists \iota. \\
&~~~~		\Send(A,m) < \Receive(B,m) \land \\
&~~~~					\Receive(B,m)	 < \Send(B,\<y, s\>) \land \\
&~~~~				 	\pred{Send} (B, \<y, s\>)	< \pred{Receive} (A, \<y, s\>)),
\end{array}
$$
where $B = \langle\agent{B}, \iota\rangle$. The bound on the security guarantee is, $\epsilon = \ever+ 2\efstwo$. Therefore $\epsilon(\eta, tb, \vec{n})  =  \epsilon^\ufcma_S(\eta, tb + \delta_s, q)+2\cdot n/\eta \cdot q_r \cdot 2^{-\eta} +
2\epsilon^\prg_P(\eta, tb + \delta_p, q')$ and $S$, $P$ are respectively the signature scheme
and pseudorandom-number generator for protocol $\Q$. 
The quantities $\delta_s$ and $\delta_p$ are the respective runtime overheads
for the proofs of the $\axname{VER}$ and $\axname{FS2}$ axioms.
These are small polynomials and are described in \ref{sec:soundness}.
Note that each role uses just one signature action. Hence $q$
in (\ref{eq:ysigned})
is upper bounded by $n_{init} + n_{resp}$, where $n_{init}$ and
$n_{resp}$ are the number of instances of the initiator and responder
roles respectively. Each role uses just one
random number generation action. Hence, $q'$ is upper bounded by $n_{init} + n_{resp}$.

\section{Related Work}
\label{sec:related}

We discuss (1) formal techniques for proving concrete security, and (2) 
formal techniques asymptotic security in computational models.

\subsection{Formal Concrete Security}

Prior work on formal methods for concrete security proofs is largely based on
reasoning about program equivalence. This approach operates at the same level
of abstraction as traditional proofs in cryptography, where security is proved
via a sequence of games, all of which are equivalent up to a small probability.
Techniques to prove equivalence of probabilistic programs have been used to
formalize cryptographic reduction proofs. Tools for assisting
development of such proofs are CryptoVerif~\cite{Blanchet06},
CertiCrypt~\cite{Barthe09certicrypt}, and EasyCrypt~\cite{EasyCrypt09}.
ZooCrypt~\cite{ZooCrypt13} contains a formal system for computing the concrete
bounds of padding-based encryption primitives. In contrast, relational reasoning
in QPCL is not exposed to the users of the proof system, and is pushed down to the soundness
proofs of axioms related to cryptography. An interesting direction to pursue would be
to mechanize the soundness proofs of QPCL's axioms using such a framework.

\subsection{Formal Asymptotic Security}

While formal techniques for concrete security proofs of cryptography are
a relatively new development, a number of approaches have been proposed
for deriving proofs of security in asymptotic models, where it
is shown that the probability with which a probabilistic polynomial
time adversary can break a scheme is a negligible function of the
security parameter. We discuss some of these approaches here.

The closest strand of related work to ours include first-order logics such as
the computational extensions of PCL~\cite{DDMST05,DDMW06} and the
logic of Bana and 
Comon-Lundh~\cite{BC-post12} (CCSA). QPCL borrows heavily from PCL's
syntax and reasoning style, which allows for the specification of precise temporal properties
and invariants of programs such as 
the matching-conversations property proved in this paper. This is an
important point 
of difference from CCSA \cite{BC-post12}, where the specification language
expresses relations between terms appearing in protocol, such as whether
one term can be derived from another. An appealing feature of CCSA is
that it is designed to be unconditionally sound, in the sense 
that all assumptions about primitives are encoded in the logic itself.
An interesting direction to pursue would be to adapt their
unconditionally sound approach to QPCL, 
while retaining the richness of our specification language.

Symbolic models have been successful
in proving the security of a large number protocols, assuming perfect
cryptography. In the so-called Dolev-Yao model~\cite{DY84}, an attacker's
capabilities are defined by a symbolic abstraction of cryptographic primitives.
The security of a protocol in a Dolev-Yao model does not, in general, imply
security under standard cryptographic assumptions.  To justify reasoning with
symbolic models, a research direction initiated by Abadi and Rogaway
\cite{AbadiRogaway}  
investigates the conditions under which security in the symbolic model implies
security in the computational model. These results (e.g.,
\cite{BPW03,MW04,Backes09cosp,CLC08,Cortier11compsoundness}), referred to 
as \emph{computational soundness proofs}, typically show that under certain
conditions a computational attacker can be translated to a symbolic attacker
with high probability.  Reasoning about cryptography in a
symbolic model 
has been shown to be helpful for automating proofs of
protocols~\cite{ProVerif,Avispa}.  Computational soundness theorems
are very strong general results that prove through complicated reductions that any program secure
in a particular symbolic model will be asyptotically secure. We
conjecture that it is because 
of the generality of these results that these
reductions have not been shown to yield precise concrete bounds for most cryptographic
schemes.

\section{Conclusion}
\label{sec:conclusion}

In this paper, we present Quantitative Protocol Composition Logic (QPCL), a
program logic for reasoning about concrete security bounds of of cryptographic
protocols. QPCL supports reasoning about temporal trace properties, Hoare
logic-style invariants and postconditions of protocol programs, and first-order
belief assertions. The semantics and soundness proofs of QPCL depend on a
formal probabilistic programming model that describes the concurrent execution
of protocol programs with a computationally bounded adversary, while accounting
for exact runtimes and probabilities. As an illustrative example, we use QPCL
to prove an authentication property with concrete bounds of a 
challeng-response protocol.
 
In future work, we hope to extend QPCL to support reasoning about non-trace
properties, such as secrecy based on computational indistinguishability, and to
cover additional cryptographic primitives (e.g., symmetric and public key
encryption, hash functions, and message authentication codes).  Because of the
correspondence between QPCL presented here and earlier versions of PCL, for
which we have proofs for deployed protocols such as SSL and Kerberos, we
believe it should be possible to develop exact security analyses for such
protocols.

\newpage
\onecolumn
\appendix
\section{Appendix}
\subsection{Proof System and Soundness Proofs}
\label{sec:soundness-app}

In this section, we prove the soundness of the axioms that we discussed in the
main text.

\begin{lemma}
The axiom \axname{PC}$_{\uparrow}$
is sound.
\end{lemma}
\begin{proof}
Fix a protocol $\Q$, and suppose that $\Qmodels
B^\epsilon \varphi$.  We must show that  $\Qmodels
\Bel{\epsilon}{\psi[P]_I\varphi}$.  Given an adversary $A$, security
parameter $\eta$, time bound $tb$,
 instance vector $\vec{n}$, and assigment $\lambda$,
let $\tree = \T_{Q, A}(\eta, tb, \vec{n})$.
Since $\Qmodels B^{\epsilon}\varphi$, we must have $
\measure{\Sem{\neg{\varphi}}^\lambda(\tree)} \leq \epsilon(\eta,tb, \vec{n})
$.

Let $\tree_1 =
\Sem{\neg(\psi[P]_I\varphi)}^\lambda(\tree) =  \{\tra \in
  \tree~|~ \exists \tra_1, \tra_2, \lambda'. (\matches{\tra,
  \lambda}{[P]_I}{\tra_1}{\tra_2}{\lambda}), \tra_1, \lambda \models
  \theta,\text{and}~\tra_2,\lambda'\not\models\phi\}$.
  We can assume without loss of generality that the free variables in
  $\varphi$ are not bound by
  $[P]_I$, since
  this can be ensured via renaming. We can also show by the definition of matching
  that $\lambda(x) = \lambda'(x)$ for all variables $x$ that are not bound by $[P]_I$.
  Therefore, we have that
  $\tree_1 = \{\tra \in
  \tree~|~ \exists \tra_1, \tra_2, \lambda'. (\matches{\tra,
  \lambda}{[P]_I}{\tra_1}{\tra_2}{\lambda}), \tra_1, \lambda \models \theta,\text{and}~\tra_2,\lambda\not\models\phi\}$.
  Additionally, we know using Lemma \ref{lemma:safety} that $\varphi$
  is a safety property. Also, by the definition
of
  matching, we know that if $(\matches{\tra,
  \lambda}{[P]_I}{\tra_1}{\tra_2}{\lambda})$, then $\tra_2 \subseteq \tra$, and consequently, $\tra_2, \lambda \not\models \varphi$
  implies $\tra, \lambda \not\models \varphi$. Therefore, we can show that:
\[
\begin{array}{lll}
  \measure{\tree_1} &=& \measure{\{ \tra \in \tree~|~ \exists \tra_1, \tra_2, \lambda'. (\matches{\tra}{[P]_I}{\tra_1}{ \tra_2}{\lambda'}), \tra_1, \lambda \models \theta,~\text{and}~\tra_{2},\lambda\not\models\phi\}} \\
                     &\leq& \measure{\{ \tra \in \tree~|~ \exists \tra_1, \tra_2, \lambda'. (\matches{\tra}{[P]_I}{\tra_1}{ \tra_2}{\lambda'}), \tra_1, \lambda \models \theta,~\text{and}~\tra,\lambda\not\models\phi\}} \\
                     &\leq& \measure{\{ \tra \in \tree~|~ \tra,\lambda\not\models\phi\}} \\
                     &=& \measure{\Sem{\neg{\varphi}}^\lambda(\tree)}\\
  &\leq& \epsilon(\eta, tb, \vec{n}).
\end{array}
\]
Therefore, $\measure{\Sem{\neg({\psi[P]_I\varphi})}^\lambda(\tree)} =
\measure{\tree_1} \leq \epsilon(\eta, tb, \vec{n})$, and
$\Qmodels \Bel{\epsilon}{\psi[P]_I\varphi}$, as desired.
\end{proof}

\begin{lemma} The axiom \axname{PC}$_{\Rightarrow}$
 is sound. \end{lemma}

\begin{proof}
We want to show that $\Qmodels (\theta
[P]_I (\varphi_1 \Rightarrow \varphi_2)) \, \wedge  \, (\theta [P]_I
\varphi_1)) \Rightarrow (\theta [P]_I \varphi_2)$.
Given $A, \eta, tb, \vec{n}, \lambda$, let $\tree = \T_{\Q,A}
(\eta, tb, \vec{n})$. Fix a trace $\tra \in \tree$.
Suppose that $\tra,\lambda \models \theta
[P]_I (\varphi_1 \Rightarrow \varphi_2)$ and  $\tra, \lambda \models
\theta [P]_I\varphi_1$. We must show that $\tra, \lambda \models
\theta [P]_I \varphi_2$. Fix $\trace_1$, $\trace_2$ and $\lambda'$
such that $\matches{\trace,\lambda}{[P]_I}{\tra_1}{\tra_2}{\lambda'}$
and $\tra_1, \lambda \models \theta$. We must show that $\tra_2, \lambda' \models \varphi_2$.
This is immediate since it must be the case that $\tra_2,\lambda' \models \varphi_1 \rimp \varphi_2$ and $\tra_2, \lambda' \models \varphi_1$.
\end{proof}

\begin{lemma} The axiom \axname{PC}$_{\land}$ is sound. \end{lemma}

\begin{proof}
This proof is analogous to the previous proof.
\end{proof}

Before proving the soundness of inference rule \axname{HON}, we state three
lemmas
asserting that the validity of a simple formula $\varphi$ on a trace is not affected if the
trace is extended with a transition in which in an action
is not executed, that is, the
configuration transition is due to one of the rules $\rulen{abort},
\rulen{wait}$, or $\rulen{switch}$. These can be
proved by a straightforward induction on the structure of the formula $\varphi$.

\begin{lemma}
\label{lemma:extabort}
If trace $\tra$ ends in configuration $\conf$, $\infern{ }{\conf \rightarrow
\conf'}{abort}$ and $\tra'$ is the extension of $\tra$ with $\conf'$, then
$\tra, \lambda \models \varphi$ implies $\tra', \lambda \models \varphi$.
\end{lemma}

\begin{lemma}
\label{lemma:extwait}
If trace $\tra$ ends in configuration $\conf$, $\infern{ }{\conf \rightarrow
\conf'}{wait}$ and $\tra'$ is the extension of $\tra$ with
$\conf'$,  then $\tra, \lambda \models \varphi$ implies $\tra',
\lambda \models \varphi$.
\end{lemma}

\begin{lemma}
\label{lemma:extswitch}
If trace $\tra$ ends in configuration $\conf$, $\infern{ }{\conf \rightarrow
\conf'}{switch}$ and $\tra'$ is the extension of $\tra$ with
$\conf'$, then $\tra, \lambda \models \varphi$ implies $\tra',
\lambda \models \varphi$.
\end{lemma}
\begin{lemma} The inference rule \axname{HON} is sound.
\end{lemma}

\begin{proof}
Suppose that $\Qmodels \pred{Start}(I) \Rightarrow \varphi$
and $\Qmodels \varphi[P]_I\varphi$ for all thread-identifiers $I$ and all roles $P
\in IS(\mathcal{Q})$.  We
must show that $\Qmodels \varphi$.
Fix $\Q, \A, \eta, tb, \vec{n}, \lambda$, and let $T = \T_{\Q,\A} (\eta, tb, \vec{n})$.
Consider a trace $\tra \in T$, a valuation $V$, and a thread-identifier $I$.
We show that $\tra, \lambda \models \varphi$ by induction on the number of
transitions in $\tra$.

For the base case, suppose that $\tra$ has no transitions (it is an empty trace or has single configuration $\conf$), then $\tra, \lambda \models
\pred{Start}(I)$.  Since  $\Qmodels \pred{Start}(I) \Rightarrow
\varphi$,
it follows that $\tra, \lambda \models \phi$, as desired.

If $\tra$ has at least one transition, then there exists some trace $\tra'$
ending with a configuration $\conf$  such that
$\tra$ is $\tra'$ extended with the transition $\conf
\rightarrow \conf'$. By the induction hypothesis, we have that $\tra' \models
\varphi$. We now have the following cases for $\conf \rightarrow \conf'$:
\begin{itemize}
	\item{($\rulen{act}$)} In this case, $\conf \rightarrow
          \conf'$ is an action transition.
Since traces are generated by executing the actions of protocol programs in
order, it can be shown that
$\matches{\tra,\lambda}{[P]_I}{\tra_1}{\tra}{\lambda'}$, for some $P \in
IS(\Q)$, and some $\tra_1 \prec \tra$. By the induction hypothesis, we know that
$\tra_1, \lambda \models \varphi$. Also, by the antecedent, we know that $\tra,
\lambda \models \varphi[P]_I\varphi$. Therefore, $\tra, \lambda' \models
\varphi$.  We can assume without loss of generality that the free variables in
$\varphi$ are not bound by $[P]_I$, since this can be ensured via renaming. We
can also show by the definition of matching that $\lambda(x) = \lambda'(x)$ for
all variables $x$ that are not bound by $[P]_I$. Therefore, we have that $\tra,
\lambda \models \varphi$.
 \item{($\rulen{abort}$)} By Lemma \ref{lemma:extabort}, $\tra, \lambda \models
\varphi$.
 \item{($\rulen{wait}$)} By Lemma \ref{lemma:extwait}, $\tra, \lambda \models
\varphi$.
 \item{($\rulen{switch}$)} By Lemma \ref{lemma:extswitch}, $\tra, \lambda \models
\varphi$.

\end{itemize}
\end{proof}

\paragraph*{Additional axioms}
In the main text, we did not give axioms for \pred{Contains},
\pred{Fresh} or equality.  We now give these axioms.  The proof of soundness for
each of these axioms is straightforward, and is left to the reader.
Axioms \axname{COMP1}-\axname{COMP7} capture computability, and \axname{NCOMP} captures non-computability. Axioms \axname{FS1}-\axname{FS5} capture freshness of nonces. Axioms \axname{EQ1} and \axname{EQ2} capture
equality.

\smallskip
\begin{description}[style=unboxed]
  \item[\axname{COMP1}.] $ \Bel{0}{ \pred{Contains} (m, m)}$
  \item[\axname{COMP2}.] $ \Bel{0}{\pred{Contains} (x, y) \land \pred{Contains} (y, z) \Rightarrow \pred{Contains} (x, z)}$
  \item[\axname{COMP3}.] $ \Bel{0}{\pred{Contains}(n, m_0) \Rightarrow \pred{Contains} (n, \<m_0, m_1\>)}$
  \item[\axname{COMP4}.] $ \Bel{0}{\pred{Contains}(n, m_1) \Rightarrow \pred{Contains} (n, \<m_0, m_1\>}$
\item[\axname{COMP5}.] $ \Bel{0}{\pred{Contains}(x, y) \ [P]_I \
\pred{Contains} (x, y)}$
		\item[\axname{COMP6}.] $ \Bel{0}{ \true \ [m' \varbind \action{sign} \ m, k]_I \ \pred{Contains} (m, m')}$
		\item[\axname{COMP7}.]  $ \Bel{0}{ \true \ [\_ \varbind \action{verify} \ m', m, k]_I \ \pred{Contains} (m', m)}$
\item[\axname{NCOMP}.]  If $m$ and $m'$ are atomic and $m \ne m'$, then
$\Bel{0}{\neg \pred{Contains} (m, m') }$.
	\item[\axname{FS1}.] $ \Bel{0}{ \true \ [n \varbind \action{new} ]_I \ \pred{Fresh} (I, n)}$
		\item[\axname{FS3}.] $ \Bel{0}{ \pred{Fresh} (I, n) \ [\action{no send actions}]_I \pred{Fresh} (I, n)}$
		\item[\axname{FS4}.] $ \Bel{0}{ \pred{Fresh} (I, n) \ [\action{send} \ m]_I \neg \pred{Contains} (n, m) \Rightarrow \pred{Fresh} (I, n)}$
		\item[\axname{FS5}.] $ \Bel{0}{ \pred{Fresh} (I, n) \ [\action{send} \ m]_I
		\pred{Contains} (n, m) \Rightarrow \pred{FirstSend}
		(I, n, m)}$
        \item[\axname{EQ1}.] $\<t_1, \cdots, t_k\> = \<t'_1, \cdots, t'_k\> \Rightarrow (t_1 = t'_1) \land \cdots \land (t_k = t'_k) $
    \item[\axname{EQ2}.] $\varphi \land (x = t) \Rightarrow [t/x]\varphi $
  \end{description}
\smallskip

Our proofs also use axioms for local reasoning about the temporal ordering
of actions in a program.
In these axioms, $a, a_1, a_2 \in \{ \action{send},
\action{receive}, \action{new}, \action{verify} \}$.

\begin{description}
  \item[\axname{AA1}.]  $\Bel{0}{\true [\cdots x \leftarrow a \ t; \cdots]_{I} \pred{a} (I, x, t)}$
  \item[\axname{AA2}.]  $\Bel{0}{\true [\cdots x_1 \leftarrow a_1 \ t_1; \cdots x_2 \leftarrow a_2
      \ t_2 \cdots]_{I}( \pred{a}_1 (I, x_1, t_1) < \pred{a}_2 (I, x_2, t_2)})$
  \item[\axname{AA3}.]  $\pred{a}_1 (I, x_1, t_1) < \pred{a}_2 (I, x_2, t_2) \Rightarrow \pred{a_1} (I, x_1, t_1)$
    \item[\axname{AA4}.]  $\pred{a}_1 (I, x_1, t_1) < \pred{a}_2 (I, x_2, t_2) \Rightarrow \pred{a_2} (I, x_2, t_2)$
    \item[\axname{AA5}.] If for all $\lambda$, $\lambda(\pred{a}(I, x, t)) \not=\lambda(\pred{a'}(I, x', t'))$, then $\neg\pred{a}(I, x, t)[x' \leftarrow a'\ t']_I\neg \pred{a}(I, x, t)$
    \item[\axname{AA6}.] $\neg\pred{a}(I, x, t)[x' \varbind a\ t']_X (\pred{a}(I, x, t) \Rightarrow (x' = x) \land (t' = t))$
    \item[\axname{Start}.] $\pred{Start}(I) \Rightarrow \neg\pred{a}(I, x, t)$
\end{description}

\noindent Informally, axiom \axname{AA1} states that that if $a$ is an action
and $\pred{a}$ is the corresponding action predicate, then after the thread with identifier $I$
executes $(x \leftarrow a \ t;)$, $\pred{a} (I, x, t)$ holds. Axiom \axname{AA2} establishes an
order between actions in a protocol program. Axioms \axname{AA3} and
\axname{AA4} state that if we can establish an order between two actions, both
actions must have actually occurred. Axiom \axname{AA5} states that if an action predicate does not unify with an action, and if the action predicate is false before the action, then it is false after the action. Axiom \axname{Start} states that no
actions with thread identifier $I$ have occurred if $\pred{Start}(I)$ holds on a trace.

\subsection{Full Proof of Protocol Example}
\label{app:example}
We now provide a full, formal proof of the correctness of the
protocol example.

Using \axname{AA1},
$$
\true [\cordname{Init}_{CR}(\agent{B})]_A \, \pred{Verify} (A, s,
    \<{\resp, y, m, \agent{A}}\>, \agent{B})
$$
is provable.  Thus, by \axname{B1} and \axname{B2} so is
$$
B^0(\true [\cordname{Init}_{CR}(\agent{B})]_A \, \pred{Verify} (A, s,
\<{\resp, y, m, \agent{A}}\>, \agent{B})).
$$
Choose $\Delta$ to be $\pred{Honest}(\agent{B})$.
By \axname{VER} instantiated to principal $\agent{B}$,
$$
\QDentails	B^\ever(
	 \pred{Verify} (A, s,  \<{\resp, y, m, \agent{A}}\>, \agent{B})\Rightarrow
	 \exists \iota. \pred{Sign} (\<\hat{B}, \iota\>, s, \<\resp, y, m,
	 \agent{A}\>) )
$$
is provable.
By \axname{PC}$_{\uparrow}$, so is
$$
			\begin{array}{l}
\QDentails	B^\ever(\true [\cordname{Init}_{CR}(\agent{B})]_A(
\pred{Verify} (A, s,  \<{\resp, y, m, \agent{A}}\>,
\agent{B}) \Rightarrow
 \exists \iota. \pred{Sign} (\<\hat{B}, \iota\>, s, \<\resp, y, m, \agent{A}\>))).
\end{array}
$$
By \axname{PC}$_{\Rightarrow}$, so is
$$
	B^{\ever}(
		 \true [\cordname{Init}_{CR}(\agent{B})]_A \ \exists \iota. \pred{Sign} (\<\hat{B}, \iota\>, s, \<\resp, y, m, \agent{A}\>))
	.
$$

Let $B$ abbreviate the term $\<\agent{B}, \iota\>$. We have that
\begin{equation}
	B^{\ever}(
		 \true [\cordname{Init}_{CR}(\agent{B})]_A \exists\iota.  \
       \pred{Sign} (B, s, \<\resp, y, m, \agent{A}\>)).
	\label{eq:ysigned1}
\end{equation}
Next, we wish to use the \axname{HON} rule to prove the following invariant $\varphi_{inv}$ about protocol
programs.
\[
  \varphi_{inv} = 	\pred{Sign} (Y, s', \<\resp, y', m',
				\agent{X}\>) \Rightarrow
         \Receive(Y, m') \land \pred{FirstSend} (Y,y', \<y', s'\>) \land \Phi'
 	)
\]
where
$$
	\Phi' \equiv \Send(Y,\<y', s'\>) \Rightarrow \Receive(Y,m') <
 				\Send(Y,\<y', s'\>).
$$

Therefore, we need to show that $\pred{Start}(X) \Rightarrow \varphi_{inv}$, and
for all initial segments $P$ of protocol roles $\varphi_{inv}[P]_X\varphi_{inv}$.
By axiom \axname{Start}, we have:

\begin{equation}
    \pred{Start}(Y) \Rightarrow \neg \pred{Sign} (Y, s', \<\resp, y', m',\agent{X}\>).
    \label{eq:HON01}
\end{equation}
And therefore, by straightforward first order reasoning,
\begin{equation}
  \pred{Start}(Y) \Rightarrow \varphi_{inv}.
    \label{eq:HON0}
\end{equation}

We need to show for every initial segment $P$ of the roles of the protocol,
$\pred{Start}(Y)[P]_Y\varphi_{inv}$. We show this for three interesting cases.

We first consider the initial segment $P_1 = \<x, \agent{X}\> \varbind \irl{receive}$.
Using \axname{AA5}, we have:

\begin{equation}
  \neg\pred{Sign}(Y, r, \<\resp, n, x, \agent{X}\>)[P_1]_Y\neg\pred{Sign}(Y, r, \<\resp, n, x, \agent{X}\>).
  \label{eq:HON1.1}
\end{equation}

By straightforward first order reasoning and using \axname{Start} and \axname{G3}, we have:
\begin{equation}
  \pred{Start}(Y)[P_1]_Y\varphi_{inv}.
  \label{eq:HON1.2}
\end{equation}

Next, we consider the initial segment $P_2 = \<x, \agent{X}\> \varbind \irl{receive}; n \varbind \irl{new}; r \varbind \irl{sign}\<\resp, n, x, \agent{X}\>$.
Using \axname{AA1}, we have:

\begin{equation}
  true[P_2]_Y\pred{Receive}(Y, x).
  \label{eq:HON2.1}
\end{equation}

Weakening using \axname{G3}, we have:
\begin{equation}
  \pred{Start}(Y)[P_2]_Y\pred{Receive}(Y, x).
  \label{eq:HON2.1.1}
\end{equation}

Also, using \axname{AA5}, we have:
\begin{equation}
  \neg\pred{Send}(Y, t')[P_2]_Y\neg\pred{Send}(Y, t').
  \label{eq:HON2.2}
\end{equation}

Also, by the definition of $\pred{FirstSend}$, we can show that:
\begin{equation}
  \neg\pred{Send}(Y, t') \Rightarrow \pred{FirstSend}(Y, n, \<n, r\>).
\label{eq:HON2.3}
\end{equation}

Therefore, using \axname{Start}, \axname{G3} and (\ref{eq:HON2.2}) we have:
\begin{equation}
  \pred{Start}(Y)[P_2]_Y\pred{FirstSend}(Y, n, \<n, r\>).
  \label{eq:HON2.4}
\end{equation}

Similarly, we can show that:
\begin{equation}
  \pred{Start}(Y)[P_2]_Y\neg\pred{Send}(Y,\langle r, n\rangle).
  \label{eq:HON2.5}
\end{equation}
And therefore, by \axname{G3},

\begin{equation}
  \pred{Start}(Y)[P_2]_Y\pred{Send}(Y,\langle r, n\rangle) \Rightarrow \pred{Receive}(Y,x) < \pred{Send}(Y,\langle r, n\rangle).
  \label{eq:HON2.6}
\end{equation}

Therefore, from (\ref{eq:HON2.2}), (\ref{eq:HON2.4}), (\ref{eq:HON2.6}), and using \axname{G1}, we have:

\begin{equation}
  \pred{Start}(Y)[P_2]_Y \pred{Receive}(Y, x) \land \pred{FirstSend}(Y, n, \langle n, r\rangle) \land \pred{Send}(Y,\langle r, n\rangle) \Rightarrow \pred{Receive}(Y,x) < \pred{Send}(Y,\langle r, n\rangle).
  \label{eq:HON2.7}
\end{equation}

Now, using \axname{AA5}, \axname{AA6}, and \axname{Start}, for fresh $u$ we can show that:

\begin{equation}
  \pred{Start}(Y)[P_2]_Y\pred{Sign}(Y, u) \rightarrow u = \<r,\<\resp, n, x, \agent{X}\>\>
  \label{eq:HON2.8}
\end{equation}

Substituting $u$ with $\<s', \<\resp, y', m', \agent{X}\>\>$, and using \axname{EQ1}, we have:

\begin{equation}
  \pred{Start}(Y)[P_2]_Y\pred{Sign}(Y, \<s', \<\resp, y', m', \agent{X}\>\>) \Rightarrow (s' = r) \land (y' = n) \land (m' = x).
  \label{eq:HON2.9}
\end{equation}

Applying \axname{G1} to \ref{eq:HON2.7} and \ref{eq:HON2.9}, and rearranging terms, we have:

\begin{equation}
  \begin{array}{l}
  \pred{Start}(Y)[P_2]_Y\pred{Sign}(Y, \<s', \<\resp, y', m', \agent{X}\>\>) \Rightarrow \\
      \quad\quad\quad (s' = r) \land (y' = n) \land (m' = x) \\
      \quad\quad\quad \pred{Receive}(Y, x) \land \pred{FirstSend}(Y, n, \langle n, r\rangle) \land \pred{Send}(Y,\langle r, n\rangle) \Rightarrow \pred{Receive}(Y,x) < \pred{Send}(Y,\langle r, n\rangle).
    \end{array}
  \label{eq:HON2.10}
\end{equation}

Using \axname{EQ2} on (\ref{eq:HON2.10}), we have what we wanted to prove.

\begin{equation}
  \begin{array}{l}
  \pred{Start}(Y)[P_2]_Y\pred{Sign}(Y, \<s', \<\resp, y', m', \agent{X}\>\>) \Rightarrow \\
      \quad\quad\quad \pred{Receive}(Y, x) \land \pred{FirstSend}(Y, s', \langle y', m'\rangle) \land \pred{Send}(Y,\langle y', s'\rangle) \Rightarrow \pred{Receive}(Y,m') < \pred{Send}(Y,\langle y', s'\rangle).
    \end{array}
  \label{eq:HON2.11}
\end{equation}

Finally, we consider the initial segment $P_3 = \<x, \agent{X}\> \varbind \irl{receive}; n \varbind \irl{new}; r \varbind \irl{sign}\<\resp, n, x, \agent{X}\>; \irl{send}\<\agent{X}, <n,r>\>$.

Similar to  (\ref{eq:HON2.1.1}), we can show that

\begin{equation}
  \pred{Start}(Y)[P_3]_Y\pred{Receive}(Y, x).
  \label{eq:HON3.1}
\end{equation}

Using \axname{FS1}, we have :

\begin{equation}
  true[n \varbind \irl{new}]_Y\pred{Fresh}(Y, n).
  \label{eq:HON3.2}
\end{equation}
We can also show using \axname{G3} that:
\begin{equation}
  \pred{Start}(Y)[\<x, \agent{X}\> \varbind \irl{new}]_Y true.
  \label{eq:HON3.3}
\end{equation}
Therefore, by \axname{S1}, (\ref{eq:HON3.2}), and (\ref{eq:HON3.3}):

\begin{equation}
  \pred{Start}(Y)[\<x, \agent{X}\> \varbind \irl{receive}; n \varbind \irl{new}]_Y \pred{Fresh}(Y, n).
  \label{eq:HON3.4}
\end{equation}

By \axname{FS3}, we have:

\begin{equation}
  \pred{Fresh}(I, n)[ r \varbind \irl{sign}\<\resp, n, x, \agent{X}\> ]\pred{Fresh}(I, n)
  \label{eq:HON3.5}
\end{equation}

Therefore, by \axname{S1}, (\ref{eq:HON3.4}), and (\ref{eq:HON3.5}):

\begin{equation}
  \pred{Start}(Y)[P_2]_Y \pred{Fresh}(Y, n).
  \label{eq:HON3.6}
\end{equation}

Also, using \axname{COMP1} and \axname{COMP3}, we have:

\begin{equation}
  \pred{Contains}(n, \<n, r\>).
  \label{eq:HON3.7}
\end{equation}

Therefore, using \axname{FS5} and (\ref{eq:HON3.7}), we have:

\begin{equation}
  \pred{Fresh}(Y, n)[\irl{send}\<n, r\>)]_Y\pred{FirstSend(Y, n, \<n, r\>)}.
  \label{eq:HON3.8}
\end{equation}

Therefore, by \axname{S1}, (\ref{eq:HON3.6}), and (\ref{eq:HON3.8}):

\begin{equation}
  \pred{Start}(Y)[P_3]_Y \pred{FirstSend}(Y, n, \<n, r\>).
  \label{eq:HON3.9}
\end{equation}

Also, by \axname{AA2}, we have:
\begin{equation}
  \pred{Start}(Y)[P_3]_Y \pred{Receive}(Y,x) < \pred{Send}(Y,\langle r, n\rangle)
\end{equation}

Therefore, we have
\begin{equation}
  \pred{Start}(Y)[P_3]_Y \pred{Send}(Y,\langle r, n\rangle) \Rightarrow
 \pred{Receive}(Y,x) < \pred{Send}(Y,\langle r, n\rangle)
  \label{eq:HON3.10}
\end{equation}

Therefore, by \axname{G1},  (\ref{eq:HON3.1}), (\ref{eq:HON3.9}), and (\ref{eq:HON3.10}).

\begin{equation}
  \pred{Start}(Y)[P_3]_Y \pred{Receive}(Y, x) \land \pred{FirstSend}(Y, n, \langle n, r\rangle) \land \pred{Send}(Y,\langle r, n\rangle) \Rightarrow \pred{Receive}(Y,x) < \pred{Send}(Y,\langle r, n\rangle).
  \label{eq:HON3.11}
\end{equation}

Analogously, to proof steps of (\ref{eq:HON2.7})-(\ref{eq:HON2.10}), we can use \ref{eq:HON3.11} to show:

\begin{equation}
  \begin{array}{l}
  \pred{Start}(Y)[P_3]_Y\pred{Sign}(Y, \<s', \<\resp, y', m', \agent{X}\>\>) \Rightarrow \\
      \quad\quad\quad \pred{Receive}(Y, x) \land \pred{FirstSend}(Y, s', \langle y', m'\rangle) \land \pred{Send}(Y,\langle y', s'\rangle) \Rightarrow \pred{Receive}(Y,m') < \pred{Send}(Y,\langle y', s'\rangle).
    \end{array}
  \label{eq:HON3.12}
\end{equation}

We can show $\pred{Start}(Y)[P]_Y\varphi_{inv}$ for every initial sequence, and by using \axname{HON},
we obtain the following invariant about protocols:

\begin{equation}
	B^{0}(
				\pred{Sign} (Y, s', \<\resp, y', m',
				\agent{X}\>) \Rightarrow
         \Receive(Y, m') \land \pred{FirstSend} (Y,y', \<y', s'\>) \land \Phi'
 	).
\label{eq:HON}
\end{equation}
By instantiating some of the free variables in (\ref{eq:HON})
and using axiom B1, we get that
$$
	B^0({
			\pred{Sign} (B, s, \<\resp, y, m, \agent{A}\>)
		 \Rightarrow \Receive(B,m) \land  \pred{FirstSend} (B,y, \<y, r\>) \land \Phi
	})
$$
is provable, where
\begin{align*}
	\Phi &\equiv \Send(B,\<y, s\>) \Rightarrow \Receive(B,m) <
			\Send(B,\<y, s\>).
\end{align*}

By \axname{PC}$_{\uparrow}$,
$$
	B^0(
 	\true [\cordname{Init}_{CR}(\agent{B})]_A \ \pred{Sign} (B, s, \<\resp, y, m, \agent{A}\>)
 \Rightarrow
		(
\Receive(B, m) \land \pred{FirstSend} (B,y, \<y, s\>) \land \Phi
)
 	)
$$
is provable.
By \axname{PC}$_{\Rightarrow}$ and (\ref{eq:ysigned1}),
$$
\QDentails	B^{\ever}({
	 	\true [\cordname{Init}_{CR}(\agent{B})]_A  \
  \exists\iota.\ (
\Receive(B, m) \land \pred{FirstSend} (B,y, \<y, s\>) \land \Phi
   )
 	})
$$
is provable.
By \axname{FS1, COMP} so is
$$
	B^{0}(\true [\cordname{Init}_{CR}(\agent{B})]_A
			(\pred{FirstSend} (A, m, m) \land \pred{Contains} (m, m)
			 \land \pred{Contains} (y,\<y,s\>))).
$$
Therefore by \axname{PC}$_{\land}$ and B1, so is
\begin{equation}
\begin{array}{l}
\QDentails	B^{\ever}(
	 \true [\cordname{Init}_{CR}(\agent{B})]_A\exists \iota.\ (
	 			\Receive(B, m) \land \pred{FirstSend} (A, m,  m) \land \pred{Contains} (m, m) \land \\
\quad\quad\quad\quad\quad\quad\quad \pred{Receive} (A, \<y,s\>) \land \pred{FirstSend}(B,y, \<y, s\>)\land
					\pred{Contains} (y, \<y,s\>) \land \Phi
	)).
\end{array}
\label{eq:hon-facts}
\end{equation}

We instantiate \axname{FS2} with message $m$ and nonce $m$, and
use this to establish $\Send(A,m) < \Receive(B,m)$:
\begin{equation}
	B^{\efstwo} (
			\pred{FirstSend} (A, m, m) \land
				\Receive(B,m) \land \pred{Contains} (m, m)
	\Rightarrow	\Send(A,m) < \Receive(B,m)
		)
\label{eq:msg-ord1}
\end{equation}

Similarly, we instantiate \axname{FS2} with message $\<y, s\>$ and nonce $s$. We will use this to establish
				$\Send(B,\<y, s\>) < \Receive(A,\<y, s\>)$.

\begin{equation}
	B^{\efstwo} (
					\pred{FirstSend} (B,
y, \<y, s\>) \land
				\Receive(B,\<y, s\>)
\land \pred{Contains} (y, \<y, s\>)
			\Rightarrow
				\Send(B,\<y, s\>) < \Receive(A,\<y, s\>)
		).
\label{eq:msg-ord2}
\end{equation}

Therefore, by the conjunction of (\ref{eq:msg-ord1}) and
(\ref{eq:msg-ord2}) and using \axname{B3}: 
\begin{equation}
				\begin{array}{l}
	B^{2\cdot\efstwo} (
		\pred{FirstSend} (A,
m, m) \land
			\Receive(B,m)
\land \pred{Contains} (m, m) \land
\pred{FirstSend} (B, y, \<y, s\>) \\
\quad\quad \land\, (\Receive(B,\<y, s\>) \land \pred{Contains} (y, \<y, s\>) \\
		\quad\quad\quad			\Rightarrow
\Send(A,m) < \Receive(B,m)) \land
			\Send(B,\<y, s\>) <
\Receive(A,\<y, s\>)
		).
				\end{array}
\label{eq:msg-ord-conj}
\end{equation}

Therefore by (\ref{eq:hon-facts}), (\ref{eq:msg-ord-conj}), and using \axname{PC}$_{\uparrow}$ and \axname{PC}$_{\Rightarrow}$,
$$
\QDentails	B^{\ever+2\cdot\efstwo}(
	 	\true [\cordname{Init}_{CR}(\agent{B})]_A \exists\iota.\
			\Send(A,m) < \Receive(B,m) \land
			\Send(B,\<y, s\>) <
\Receive(A,\<y, s\>) \land \Phi
		)
	).
$$
Finally, by $\axname{AA4}$ and $\axname{PC}_{\Rightarrow}$ on $\Phi$,
$$
\begin{array}{ll}
\QDentails	B^{\ever+2\cdot\efstwo}(
	 	\true [\cordname{Init}_{CR}(\agent{B})]_A \exists\iota.\  (
\Send(A,m) < \Receive(B,m) \\ \quad \quad \land 
		\Receive(B,m) < \Send(B,\<y,
s\>) \land
		\Send(B,\<y, s\>) < \Receive(A,\<y,
 s\>)
		)
	).
			\end{array}
$$

\newcommand{\V}{{\cal V}}

\subsection{A First-Order Logic of Belief}
\label{section:conditionallogic}

Halpern \cite{Hal37} presents a ``quantitative'' variant of first-order
conditional logic that includes formulas whose semantics capture the
form of conditional probability statements that are common in cryptographic
security definitions.
In this paper, we use a simplified version of that logic as a foundation,
and show how we can
devise an elegant, powerful variant of the logic considered in
\cite{DDMST05}. We summarize below the relevant results of \cite{Hal37}.

The syntax of (qualitative) first-order conditional logic is straightforward.
Fix a finite first-order vocabulary $\V$ consisting, as usual,
of function
symbols, predicate symbols, and constants.  Starting with atomic
formulas of first-order logic over the vocabulary $\V$, we form more
complicated formulas by closing off under the standard truth-functional
connectives
(i.e., $\land$ ,$\lor$, $\neg$, and $\rimp$), first-order
quantification,
and the binary modal operator $\Cond$.  Thus, a typical formula is
$\forall x (P(x) \Cond \exists y (Q(x,y) \Cond R(y)))$.
Let $\LCondfo(\V)$ be the resulting language.
We are most interested in the fragment $\LCondfo^0(\V)$ of
$\LCondfo(\V)$ consisting of all formulas in $\LCondfo(\V)$
that have the form
$
\forall x_1 \ldots \forall x_n ( \phi \Cond \psi)$, where
$\phi$ and $\psi$ are first-order formulas.
(We henceforth omit the $\V$ unless it is necessary for clarity.)
Note that in $\LCondfo^0$ we cannot negate an $\rarrow$
formula, nor can we take the conjunction of two $\rarrow$ formulas.
As we shall
see, the fact that we cannot take the conjunction is not a real lack of
expressive power, since we consider sets of formulas (where a set can be
identified with the conjunction of the formulas in the set).
However, the lack
of negation does make for a loss of expressive power.

There are many approaches to giving semantics to first-order
conditional logic (see \cite{Hal31} for an overview).  The one most
relevant here is what has been called $\epsilon$-semantics \cite{GMPfull}.
Under this approach, 
the semantics of formulas in $\LCondfo(\V)$ is given with respect to
\emph{PS structures}.
A PS structure is a tuple $M = (D, W,\pi, \P)$, where $D$ is a domain,
$W$ is a set of worlds, $\pi$ is an \emph{interpretation} that
associates with each predicate symbol (resp., function symbol,
constant symbol) in $\V$ and world $w \in W$ a predicate (resp.,
function, domain element) of the right arity, and $\P = \<\Pr_1,
\Pr_2, \ldots\>$ is a probability sequence.  As
usual, a \emph{valuation} $V$
associates with each variable $x$ an element $V(x) \in D$.

Given a valuation $V$, we give semantics to $\land$, $\neg$, $\rimp$,
and
$\forall$ in the standard way.  In particular, we determine whether a
first-order formula $\phi$ is true at a world $w \in W$ as usual.  In
this case, we write   $(M,V,w) \satp \phi$.  Let
$\intension{\phi}_{M,V} = \{w: (M,V,w)
\satp\phi\}$.
If $\phi$ is a closed formula, so that its truth does
not depend on the valuation, we occasionally write $\intension{\phi}_M$
rather than $\intension{\phi}_{M,V}$.
We write $(M,V) \satp \phi$ if $(M,V,w) \satp \phi$ for all worlds $w$.

The truth of an $\rarrow$ formula does not depend on
the world, but only on the PS structure.
We can consider two possible semantics for conditional formulas in PS
structures.  The first just asks for convergence: 
$$(M,V) \satp \phi \Cond \psi \mbox{ if } \lim_{n \rightarrow \infty} {\Pr}_n
(\intension{\psi}_{M,V}  \mid \intension{\phi}_{M,V}) = 1,$$
where ${\Pr}_n (\intension{\psi}_{M,V}  \mid \intension{\phi}_{M,V})$ is
taken to be 1 if ${\Pr}_n(\intension{\phi}_{M,V}) = 0$.

The second approach may be of more interest in cryptographic applications, and
asks for super-polynomial convergence.
In this approach, we define
\begin{itemize}
\item[]
$(M,V) \satp \phi \Cond \psi$ if for all $k \ge 1
$, there exists some $n_k \ge 0$
such that, for all $n \ge n_k$,
${\Pr}_n (\intension{\psi}_{M,V}  \mid \intension{\phi}_{M,V}) \ge 1 -
(1/n^k)$.\footnote{Note that it  follows that if $p(n)$ is a polynomial whose leading
coefficient is positive, then ${\Pr}_n (\intension{\psi}_{M,V}  \mid
\intension{\phi}_{M,V}) \ge 1 - (1/p(n))$ for sufficiently large $n$.
In \cite{Hal37}, this was the condition used to define the semantics of
$\phi \Cond \psi$.  The approach used here is simpler, but equivalent.}
\end{itemize}

It turns out that the two appraoches are characterized by the same
axiom system, so we do not distinguish them further here.
As usual, we write $M \satp \phi$ if $(M,V) \satp \phi$ for all
valuations $V$, and $\M \satp \phi$ if $M \satp \phi$ for all PS
structures in a set $\M$.

In this paper, we need only formulas where the antecedent of
$\Cond$ is $\true$.  We take $B \phi$ (which can be read ``$\phi$ is
believed'' or ``the agent believes $\phi$') to be an abbreviation of $\true
\rarrow \phi$.  Roughly speaking, $B \phi$ holds if the probability of
$\phi$ approaches 1 faster than any inverse polynomial.  Let
$\LCondfo^B$ be the fragment of $\LCondfo^0$ where $\Cond$ occurs only
in the context $B$.
As we have seen, such formulas are used in \axname{VER}; they are also
needed in, for example, the CMA security definition for digital
signatures~\cite{BR-notes01}).

Sound and complete axiomatizations for first-order conditional logic
are well known \cite{FHK}.
Halpern \citeyear{Hal37} gives a sound and complete axiomatization for
$\LCondfo^0$.  The advantage of considering this fragment is that it
lends itself naturally to quantitative reasoning.  
We now review the axiomatization of $\LCondfo^0$, show how to
specialize it to $\LCondfo^B$,
and give the quantitative version of the axioms.

\newcommand{\Lfo}{\L^{\mathit{fo}}} 
\newcommand{\LCondBq}{\LCondfo^{B,q}}
\newcommand{\PfoB}{\mathbf{P}^{B}_{\Lambda}}
\newcommand{\PfoBq}{\mathbf{P}^{B,q}_{\Lambda}}
\newcommand{\Pfopq}{\mathbf{P}_\Lambda^{+,q}}
\newcommand{\inter}{\cap}

\subsection{A Sound and Complete Axiomatization for First-Order Conditional
Logic}
\label{sec:focl-details}

To describe the axioms, we define two more languages.  Let $\Lfo$
be the pure first-order fragment of $\LCondfo$, consisting of
$\rarrow$-free formulas;
let $\LCondfoo$ be the fragment of $\LCondfo^0$ consisting
of all formulas of the form $\phi \Cond \psi$ where $\phi$ and $\psi$
are closed first-order formulas.   The axiomatization is based on system
$\sysPL$, also known as \emph{the KLM properties} \cite{KLM}, restricted
to models where all worlds satisfy a particular first-order theory
$\Lambda$.  The system $\sysPL$ applies only to formulas in $\LCondfoo$
(so $\rarrow$ formulas cannot be universally quantified); we will need
additional axioms to deal with quantification.

We state the axioms in terms of statements of the form
$\Delta\hra \phi$, where $\Delta \subseteq \Lfo \union \LCondfo^0$
and $\phi \in \LCondfo^0$.  Roughly speaking, $\Delta\hra \phi$ is
interpreted as ``$\phi$ follows from the formulas in $\Delta$''.
We write $(M,V) \satp \Delta \hra \phi$ if $(M,V) \satp \phi' $
for every formula $\phi' \in \Delta$ implies that $(M,V)
\satp \phi$. A collection $\Delta$ of formulas in $\LCondfo^0$ can be
written as $\Deltar \union \Deltafo$, where 
$\Deltar \subseteq \LCondfo^0$ and $\Deltafo \subseteq \Lfo$.

We first describe the axioms of $\sysPL$, adapted to this framework.
The axioms all have conclusions of the form $\phi \Cond \psi$, since
they are intended for reasoning about $\LCondfoo$.
Let $\vdashL$ denote provability in first-order logic given the axioms
in the theory $\Lambda$.
(In describing these axioms, we keep the naming convention used in
\cite{Hal37}, which in turn is based on \cite{KLM}, for consistency.)

\begin{itemize}
  \item[LLE$^+$.]
If $\vdash_{\Lambda \union \Deltafo} \phi_1 \dimp \phi_2$,
then from $\Delta \hra \phi_1 \Cond \psi$ infer $\Delta \hra \phi_2
\Cond \psi$ (left logical equivalence).
  \item[RW$^+$.]
If $\vdash_{\Lambda \union \Deltafo} \psi_1 \rimp \psi_2$, then from
  $\Delta \hra \phi\Cond\psi_1$ infer
     $\Delta \hra \phi\Cond\psi_2$ (right weakening).
  \item[REF.] $\Delta \hra \phi\Cond\phi$ (reflexivity).
  \item[AND.] From $\Delta \hra \phi\Cond\psi_1$ and $\Delta \hra \phi\Cond\psi_2$ infer
     $\Delta \hra \phi\Cond \psi_1 \land \psi_2$.
  \item[OR.] From $\Delta \hra \phi_1\Cond\psi$ and $\Delta \hra
\phi_2\Cond\psi$ infer
     $\Delta \hra \phi_1\lor\phi_2\Cond \psi$.
  \item[CM.] From $\Delta \hra \phi_1\Cond\phi_2$ and $\Delta \hra \phi_1\Cond\psi$ infer
     $\Delta \hra \phi_1 \land \phi_2 \Cond \psi$ (cautious monotonicity).
\end{itemize}
It is interesting to contrast LLE$^+$, RW$^+$, and CM.
While RW$^+$ allows a formula $\psi_1$ on the right-hand side of
$\Cond$ to be replaced by a weaker formula $\psi_2$ (that is, a
formula such that that $\psi_1 \rimp \psi_2$ is provable), LLE$^+$
just allows a formula $\phi_1$ on the left-hand side of $\Cond$ to be
replaced by an equivalent formula $\phi_2$, rather than a stronger
formula.
CM allows the replacement of a formula $\phi_1$ on the left-hand side
by a stronger formula, $\phi_1 \land \phi_2$, but only if $\phi_1
\Cond \phi_2$ holds.
Intuitively, this says that if $\phi_2$ and $\psi$ each almost always
hold given $\phi_1$, then $\psi$ almost always holds given both
$\phi_1$ and $\phi_2$.
Monotonicity does not hold in general.
That is, if $\phi_1 \rimp \phi_2$ is provable and $\phi_2 \Cond
\psi$ holds, then $\phi_1 \Cond \psi$ does not necessarily hold.
For a simple
counterexample, it is not the case that if $\true \Cond \psi$ holds
then $\neg
\psi \Cond \psi$ holds.
If $\psi(x)$ states that $x$ cannot forge a signature, we might expect
that almost always, $x$ cannot forge the signature ($\true
\Cond \psi(x)$), but it surely is not the case that $x$ cannot forge
the signature given that $x$ can forge it.
The other two axioms AND and OR are easy to explain. AND says
that if both $\psi_1$ and $\psi_2$ almost always hold given $\phi$,
then so does $\psi_1 \land\psi_2$, while OR allows reasoning by cases:
if $\psi$ almost always holds given each of $\phi_1$ and $\phi_2$,
then it almost always holds given their disjunction.

The following additional axioms account for universal quantification,
and the fact that $\LCondfo^0$ includes  first-order formulas.
One of the axiom uses the notion of \emph{interpretation-independent}
formula, which is a first-order formula that does not mention any
constant, function, or predicate symbols (and thus, is a formula whose
atomic predicates are all of the form $x=y$).
In the axioms below, we assume that $\Lambda$ is a first-order theory.

\begin{itemize}
\item[$\Lambda$-AX$^+$.] If $\phi \in \Lfo$ and $\vdash_{\Lambda \union
\Deltafo} \phi$,
then from  $(\Delta \union \{\phi\}) \hra \psi$ infer
$\Delta \hra \psi$.\footnote{This rule has changed somewhat from the one in \cite{Hal37},
to make it more convenient to apply.}
\item[F1$^+$.] If $z$ is a variable that does not appear in $\phi$, then
from $(\Delta \union \phi[x/z]) \hra \psi$ infer
$(\Delta \union \forall x \phi) \hra \psi$.
\item[F3$^+$.] If $x$ does not appear free in $\Delta$, then from
$\Delta \hra \phi$ infer $\Delta \hra \forall x \phi$.
\item[EQ.] If $x$ does not appear free in $\Delta$, $\phi$, or $\psi$,
and $\sigma$ is a first-order formula, then
from $(\Delta \union \{\sigma\}) \hra \phi$ infer
$(\Delta \union \{\exists x \sigma\}) \hra \phi$
(existential quantification).
\item[REN.] If $y_1, \ldots, y_n$ do not appear in $\phi$,
then from  $\Delta \hra \forall x_1 \ldots \forall x_n \phi$ infer
$\Delta \hra \forall y_1 \ldots \forall y_n (\phi[x_1/y_1, \ldots, x_n/y_n])$
(renaming).
\item[II.] If $\sigma_1$ and $\sigma_2$ are interpretation-independent,
then from $(\Delta \union \{\sigma_1\}) \hra \phi$ and
$(\Delta \union \{\sigma_2\}) \hra \phi$, infer 
$(\Delta \union \{\sigma_1 \lor \sigma_2\}) \hra \phi$ 
(interpretation independence).
\end{itemize}

Let $\Pfop$ be the axiom system consisting of LLE$^+$, RW$^+$, REF, AND, OR,
CM, $\Lambda$-AX$^+$, F1$^+$, F3$^+$, EQ, REN, and II.
A \emph{derivation from $\Pfop$} consists of a
sequence of steps, where each step has the form $\Delta \hra \phi$ and
either (a) $\phi \in \Delta$ or
(b) there is an inference rule ``if condition $C$ holds, then from
$\Delta_1 \hra \phi_1$, \ldots,  $\Delta_k \hra \phi_k$ infer
$\Delta \hra \phi$'', condition $C$ holds, and
$\Delta_1 \hra \phi_1, \ldots, \Delta_k \hra \phi_k$ appear in
earlier steps of the derivation.  (The ``condition $C$'' might be a statement
such as 
``$\vdash_{\Lambda \union \Deltafo} \phi \dimp \psi$ and $\phi \in
\Lfo$'', which appears in
LLE$^+$.)
We write $\Pfop \vdash \Delta \hra \phi$ if there is a derivation from
$\Pfop$ whose last line in $\Delta \hra \phi$.

Let $\PS(\Lambda)$ consist of all PS structures $M$ 
such that world in $M$ satisfies $\Lambda$.
We write $\PS(\Lambda)
\satp \Delta \hra \phi$ if $(M,V) \satp
\Delta \hra \phi$ for all PS structures $M \in \PS(\Lambda)$
and valuations $V$.

\begin{theorem}\cite{Hal37}\label{thm:soundandcomplete}
If $\Delta \union \{\forall x_1 \ldots \forall x_n(\phi \Cond \psi)\}
\subseteq \LCondfo^0$, then
$\PfoB \vdash \Delta \hra \forall x_1 \ldots \forall
x_n
(\phi \Cond \psi)$ if and only if
$\PS(\Lambda) \satp \Delta \hra \forall x_1 \ldots
\forall x_n(\phi \Cond \psi)$.
\end{theorem}

We now specialize these axioms to $\LCondfo^B$.  We start with the
analogues of axioms from $\sysPL$.  Consider the following three axioms:

\begin{description}
\item[B1.] If $\vdash_{\Lambda \union \Deltafo} \phi \rimp \phi'$,
then from $\Delta \hra B\phi$ infer $\Delta \hra B\phi'$.
\item[B2.] $\Delta \hra B (\true)$.
\item[B3.] From $\Delta \hra B\phi$ and $\Delta \hra B\phi'$ infer
$\Delta \hra  B(\phi \land \phi')$.
\end{description}
It should be clear that B1 is just RW$^+$ restricted to $\LCondfo^B$;
B2 is REF restricted to the only formula in $\LCondfo^B$
to which it applies, namely $\true \Cond \true$; and B3 is the
restriction of AND to $\LCondfo^B$.  It is easy to see that there are no
instances of LLE$^+$, OR, or CM in $\LCondfo^B$.
The version of B1--B3 in the main text can be viewed as the special case
where $\Delta = \emptyset$.

Let $\PfoB$ be the result of replacing LLE$^+$, RW$^+$, REF, AND, OR,
and CM in $\Pfop$ by B1--B3.
Given a collection $\Delta$ of $\rarrow$ formulas, we write
$\PfoB \vdash \Delta \hra \phi$ if $\phi\Cond\psi$ can be
derived from $\Delta$ using these rules.
A \emph{derivation from $\Delta$} consists of a sequence of steps of the
form  
$\Delta \hra \phi$, where either (a) $\phi \in
\Delta$, (b) $\phi$ is $B_i(\true)$ (so that this is just axiom B2), or
(c) $\Delta \hra \phi$ follows from previous steps by an application of
B1 or B3.

\begin{theorem}
If $\Delta \subseteq \LCondfo^B \union \Lfo$
and $ \forall x_1 \ldots \forall x_n B\phi   \in
\LCondfo^B$,
 then
$\PfoB \vdash \Delta \hra \forall x_1 \ldots \forall x_n
B \phi$ if and only if
$\PS(\Lambda) \satp \Delta \hra \forall x_1 \ldots
\forall x_n B\phi$.
\end{theorem}

\begin{proof}  The soundness follows from the soundness of $\Pfop$,
proved in \cite{Hal37}.  To prove completeness, we first prove that
B1-B3 are sound and complete for the fragment 
$\LCondfo^B \inter \LCondfoo$. 
Again, soundness is immediate.  Completeness is also quite
straightforward,
as we now show.

We start by showing that $\PS(\Lambda) \satp \{\psi_1, \ldots, \psi_k,
B\phi_1, \ldots, B\phi_m\} \hra B\sigma$ (where  $\psi_1, \ldots,
\psi_k, \phi_1, \ldots, \phi_m$, and $\sigma$ are all in $\Lfo$) iff
$\Lambda \satp (\psi_1 \land \ldots \land \psi_k \land \phi_1 \land
\ldots \land \phi_m) \rimp \sigma$.
The ``if''
direction is obvious.  For the ``only if'' direction, suppose that
$\Lambda \not\satp (\psi_1\land \ldots \land \psi_k \land \phi_1 \land
... \land \phi_m) \rimp \sigma$.
Then $\sigma' = \psi_1 \land \ldots \land \psi_k \land \phi_1 \land
\ldots \land \phi_m \land \neg \sigma$ is satisfied in a relational
structure $R$ that also satisfies
all the formulas in 
$\Lambda$.   Let $M = (D,\{w\},\pi, \P)$, where $D$ is the domain in
$R$, $\pi(w)$ is the interpretation in $R$, and 
$\P = (\Pr, \Pr, \Pr, \ldots)$, and $\Pr$ assigns probability 1 to world
$w$.  Then $M \in \PS(\Lambda)$ and $M \satp \psi_1 \land \ldots \land
\psi_k \land B\phi_1
\land \ldots \land B\phi_m$, but $M \satp \neg B \sigma$.  Thus,
$M \not\satp \{\psi_1, \ldots, \psi_k, B\phi_1, \ldots, B\phi_m\} \hra
B\sigma$.

Now suppose that $\PS(\Lambda) \satp \Delta \hra B\phi$, where $\Delta =
\{\psi_1, \ldots, \psi_k, B\phi_1, \ldots, B\phi_m\}$.
Since $\vdash_{\Lambda \union \Deltafo} \true \rimp (\psi_1
\land \ldots \land \psi_k)$, by B1 and B2, it follows that
$\PfoB \vdash \Delta \hra B(\psi_1 \land \ldots \land \psi_k)$.
Since $B\phi_1, \ldots, B\phi_m \in \Delta$, applying B3 repeatedly, it
follows that $\PfoB \vdash B(\psi_1 \land \ldots \land \psi_k \land
\phi_1 \land \ldots  \land \phi_m)$.
By the arguments above and the completeness of first-order logic,
since $\PS(\Lambda) \satp \Delta \hra B\phi$,
we must have  $\vdash_{\Lambda} (\psi_1 \land \ldots \land \psi_k
\land \phi_1 \land \ldots \land \phi_m) \rimp \sigma$.  Thus, applying
B1 again, it follows that $\PfoB \vdash \Delta \hra B \phi$.

Given the soundness and completeness of B1-B3 for
the $\LCondfoo \cap \L^B$, the proof that adding the remaining
axioms gives soundness and completeness $\L^B$ is now identical to that
going from $\LCondfoo$ to $\LCondfo$ given in Theorem 3.5 of
\cite{Hal37}. We omit further details here.
\end{proof}

\subsection{Quantitative reasoning}
Up to now, we have just considered a qualitative semantics for
conditional logic.  A  formula of the form $\phi \Cond \psi$ is true if,
for all $k$, the conditional probability of $\psi$ given $\phi$ 
is at least  $1  - 1/n^k$ for all sufficiently large $n$.
Here $n$ can be, for example, the security parameter.  While this
asymptotic complexity certainly gives insight into the security of a
protocol, in practice, a system designer wants to achieve a certain
level of security, and needs to know, for example, how large to take the
keys in order  to achieve this level.  In \cite{Hal37}, a more
quantitative semantics appropriate for such reasoning was introduced
and related to the more qualitative ``asymptotic'' semantics.  We
briefly review the relevant details here, and specialize them to the
language of belief.

The syntax of the quantitative language considered in \cite{Hal37},
denoted $\LCondfoq$, is 
just like that of $\LCondfo^0$, except that instead of formulas of the
form $\phi \Cond \psi$ there are  formulas of the form $\phi \Cond^r
\psi$, where
$r$ is a real number in $[0,1]$.  The
semantics of such a formula is straightforward:
$$\begin{array}{ll}
(M,V) \models \phi \Cond^r \psi \mbox{ if  there exists some $n^* \ge
0$ such that}\\
\mbox{for all $n
\ge n^*$, } {\Pr}_n(\intension{\psi}_{M,V}  \mid \intension{\phi}_{M,V}) \ge 1-
r).
\end{array}$$
In our quantitative logic, we do not have a sequence of probability
distributions on the
left hand-side of $\satp$; rather, we have a single distribution
(determined by the function $\epsilon$ and the parameters $\eta$ $tb$,
and $\vec{\eta}$).  Nevertheless, as we shall see, the axioms that
characterize quantitative conditional logic in the case of a sequence
of distributions are essentially the same as those for a single distribution.

Let $\LCondfoq^0$ be the obvious analogue of $\LCondfo^0$, consisting of
all formulas of the form $\forall x_1 \ldots \forall x_n (\phi \Cond^r
\psi)$, where $\phi$ and $\psi$ are first-order formulas.
For each of the axioms and rules in system $\sysPL$, there is a
corresponding sound axiom or rule in $\LCondfoq^0$:
\begin{itemize}
 \item[LLE$^q$.]
If $\vdash_{\Lambda \union \Deltafo} \phi_1 \dimp \phi_2$, then from
 $\Delta \hra \phi_1\Cond^r\psi$ infer
    $\Delta \hra \phi_2\Cond^r\psi$.
 \item[RW$^q$.]
If $\vdash_{\Lambda \union \Deltafo} \psi_1 \rimp \psi_2$,
then from
 $\Delta \hra \phi\Cond^r \psi_1$ infer
    $\Delta \hra \phi\Cond^r\psi_2$.
 \item[REF$^q$.] $\Delta \hra \phi\Cond^0 \phi$ (reflexivity).
\item[AND$^q$.] From $\Delta \hra \phi\Cond^{r_1}\psi_1$ and
$\Delta \hra \phi\Cond^{r_2}\psi_2$ infer $\Delta \hra \phi\Cond^{r_3}
\psi_1 \land \psi_2$,
where $r_3 = \min(r_1 + r_2, 1)$.
\item[OR$^q$.] From $\Delta \hra \phi_1\Cond^{r_1}\psi$ and
$\Delta \hra \phi_1\Cond^{r_2}\psi$ infer $\Delta \hra \phi_1 \lor \phi_2
\Cond^{r_3} \psi$,
where $r_3 = \min(\max(2r_1, 2r_2),1)$.
\item[CM$^q$.] From $\Delta \hra \phi_1\Cond^{r_1}\phi_2$ and
$\Delta \hra \phi_1\Cond^{r_2}\psi$ infer
    $\Delta \hra \phi\land \phi_2 \Cond^{r_3} \psi$, where $r_3 =
\min(r_1+r_2,1)$.
\end{itemize}
Let $\Pfopq$ consist of the rules above, together with F1$^+$, F3$^+$, EQ,
REN, and II (all of which hold with no change in the quantitative
setting), and
\begin{itemize}
\item[INC.] If $r_1 \le r_2$, then from $\Delta \hra \phi \Cond^{r_1} \psi$
infer $\Delta \hra \phi \Cond^{r_2} \psi$.
\end{itemize}

\begin{theorem}\cite{Hal37} The rules in $\Pfopq$ are all sound.
\end{theorem}

There is a deep relationship between $\Pfop$
and $\Pfopq$.  To make it precise, given a set of formulas $\Delta
\subseteq \LCondfo^0$, say that $\Delta' \subseteq \LCondfoq^0$ is a
\emph{quantitative instantiation} of $\Delta$ if, for every formula
$\phi \Cond \psi \in \Delta$, there is a bijection $f$ from $\Delta$ to
$\Delta'$ such that, for every formula $\phi \Cond \psi \in \Delta$,
there is a real number $r \in [0,1]$ such that $f(\phi \Cond \psi) =
\phi \Cond^r \psi$.  That is, $\Delta'$ is a quantitative instantiation
of $\Delta$ if each qualitative formula in $\Delta$ has a quantitative
analogue in $\Delta'$.

The key result in \cite{Hal37} is the following, which shows
the power of using of $\Pfop$.  Specifically, it shows
that if $\Delta \hra \phi \Cond \psi$ is derivable in
$\Pfop$ then, for all $r \in [0,1]$, there exists
a quantitative instantiation $\Delta'$ of $\Delta$ such that $\phi
\Cond^r \psi$ is derivable from $\Delta'$ in $\Pfopq$.  Thus, if the
system designer wants security at level $r$ (that is, she wants to know
that the desired security property holds with probability at least
$1-r$), then if she has a qualitative proof of the result, she can
compute the strength with which her assumptions must hold in order for
the desired conclusion to hold.  For example, she can compute how to set the
security parameters in order to get the desired level of security.  This
result can be viewed as justifying qualitative reasoning.  Roughly
speaking, it says that it is safe to avoid thinking about the
quantitative details, since they can always be derived later.
Note that this result
would not hold if the language allowed negation.  For example, even
if $\neg (\phi \Cond \psi)$ could be proved given some assumptions
(using the axiom system $\AXCL$), it would not necessarily follow that
$\neg (\phi \Cond^q \psi)$ holds, even if the probability of the
assumptions was taken arbitrarily close to one.

\begin{theorem}\label{thm:quantitative} \cite{Hal37} If $\Pfop \vdash
\Delta \hra \phi \Cond
\psi$, then for all $r \in [0,1]$, there exists a quantitative
instantiation $\Delta'$ of $\Delta$ such
that $\Pfopq \vdash \Delta' \hra \phi \Cond^q \psi$.  Moreover,
$\Delta'$ can be found in polynomial time, given the derivation of
$\Delta \hra \phi \Cond \psi$.
\end{theorem}

Intuitively, Theorem~\ref{thm:quantitative} says that if there is a
qualitative proof of $\phi \rarrow \psi$ from the assumptions in
$\Delta$, given a desired confidence level $1-r$, 
we can, in polynomial time, compute the
probability with which the assumptions in $\Delta$ need to hold in order
to ensure that the conclusion $\phi \rarrow \psi$ holds with confidence
$1-r$.  

Let $\LCondBq$ be the restriction of $\LCondfoq$ so that $\Cond^r$
appears only in the context of $\true \Cond^r \phi$, which we abbreviate
as $B^r \phi$.  The quantitative analogue of B1--B3 is obvious:
\begin{description}
\item[B1$^q$] If $\vdash_{\Lambda \union \Deltafo} \phi \rimp \phi'$,
then from $\Delta \hra B^r\phi$ infer $\Delta \hra B^r\phi'$.
\item[B2$^q$] $\Delta \hra B^0(\true)$.
\item[B3$^q$] From $\Delta \hra B^{r_1}\phi$ and $\Delta \hra
B^{r_2}\phi'$ infer 
$\Delta \hra  B^{r_3}(\phi \land \phi')$, where $r_3 = \min(r_1 + r_2, 1)$.
\end{description}

Let $\PfoBq$ consist of B1$^q$, B2$^q$, B3$^q$,F1$^+$, F3$^+$, EQ,
REN, II, and INC.  
The key reason that a sequence of probability distributions was used
even for the quantitative semantics in \cite{Hal37} is that it made it
possible to establish an important connection between the qualitative
and quantitative semantics.  
The same argument as that used in \cite{Hal37} for
$\Pfoq$ can be used to show the following result:

\begin{theorem} If $\PfoB \vdash \Delta \hra B \psi$, then for all $r
\in [0,1]$, there exists a quantitative
instantiation $\Delta'$ of $\Delta$ such
that $\PfoBq \vdash \Delta' \hra B^q \psi$.  Moreover,
$\Delta'$ can be found in polynomial time, given the derivation of
$\Delta \hra \phi \Cond \psi$.
\end{theorem}
Although we have stated all the axioms above in terms of constants $r$,
$r_1$, $r_2$, and $r_3$, we can easily modify them so as to state them
in terms of functions $\epsilon(\eta, tb, \vec{n})$; we leave details to
the reader.

This result is quite relevant for us.  It says that if we an prove a
statement in qualitative conditional logic, then for all desired
probability levels $1-q$, in all models where
the assumptions hold with sufficiently high probability, the conclusion
will hold with probability at least $1-q$.

\paragraph{Acknowledgements}
We thank Riccardo Pucella, Bob Harper and our anonymous reviewers for
their comments and suggestions.
The authors are supported in part by the AFOSR MURI on Science of Cybersecurity
(FA9550-12-1-0040), and the NSF TRUST Center (CCF 0424422).
Halpern's work is supported in part by NSF grants
IIS-0812045, IIS-0911036, and CCF-1214844, by AFOSR grants
FA9550-08-1-0438, FA9550-09-1-0266, and FA9550-12-1-0040, and by ARO
grant W911NF-09-1-0281.

\bibliographystyle{IEEEtran}
\bibliography{protocols,joe,z}
\end{document}